\documentclass[11pt]{article}

\usepackage[margin=1in,letterpaper]{geometry}
\usepackage[T1]{fontenc}
\usepackage{lmodern}
\usepackage{amsmath}
\usepackage{amssymb}
\usepackage{amsthm}
\usepackage{color}
\usepackage{microtype}
\usepackage{verbatim}
\usepackage{xspace}
\usepackage{enumitem}
\usepackage[usenames,dvipsnames]{xcolor}
\usepackage[colorlinks=true,allcolors=Maroon]{hyperref}
\usepackage{algpseudocode}
\usepackage{graphicx}

\newtheorem{theorem}{Theorem}[section]
\newtheorem{lemma}[theorem]{Lemma}

\newtheorem{proposition}[theorem]{Proposition}

\newtheorem{definition}[theorem]{Definition}
\newtheorem{problem}{Problem}

\newcommand{\R}{\mathbb{R}}

\newcommand{\Rbb}{\mathbb{R}}
\newcommand{\eps}{\varepsilon}
\newcommand{\calD}{\mathcal{D}}

\newcommand{\Z}{\mathbb{Z}}
\newcommand{\calL}{\mathcal{L}}

\DeclareMathOperator{\poly}{poly}

\begin{document}
\title{Approximate Near Neighbors for General Symmetric Norms}
\author{Alexandr Andoni\\{\small Columbia University}\\{\small
    \texttt{andoni@cs.columbia.edu}} 
\and
Huy L. Nguyen\\{\small Northeastern University}\\{\small
    \texttt{hu.nguyen@northeastern.edu}} 
\and
  Aleksandar Nikolov\\ {\small University of Toronto}\\{\small \texttt{anikolov@cs.toronto.edu}}
  \and Ilya Razenshteyn\\{\small MIT CSAIL}\\{\small \texttt{ilyaraz@mit.edu}} \and
  Erik Waingarten\\{\small Columbia University}\\{\small \texttt{eaw@cs.columbia.edu}}}
\maketitle
\begin{abstract}
We show that every \emph{symmetric} normed space admits an efficient nearest
neighbor search data structure with \emph{doubly-logarithmic} approximation.
Specifically, for every $n$, $d = n^{o(1)}$, and every $d$-dimensional symmetric norm $\|\cdot\|$,
there exists a data structure for $\poly(\log \log n)$-approximate nearest neighbor search over
$\|\cdot\|$ for $n$-point datasets achieving $n^{o(1)}$ query time
and $n^{1+o(1)}$ space. The main technical ingredient of the algorithm is
a low-distortion embedding of a~symmetric norm into a low-dimensional
iterated product of top-$k$ norms.

We also show that our techniques cannot be extended to \emph{general} norms.
\end{abstract}

\thispagestyle{empty}
\pagebreak
\setcounter{page}{1}

\section{Introduction}

The Approximate Near Neighbor problem (ANN) is defined as follows.
The input is a dataset $P$ lying in a metric space $(X, d_X)$, a
distance threshold $r > 0$, and a desired approximation $c > 1$.  The
goal is to preprocess $P$ so that, given a query point $q \in X$, with
the promise that at least one of the data points is within distance $r$,
output a data point within distance $cr$ from $q$.
The ANN problem is an important tool in modern data analysis, and, at the
same time, is a source of many exciting theoretical developments, see,
e.g., the survey in~\cite{andoni16-bigData}.

In many applications, the metric is defined on $d$-dimensional real
vectors $\Rbb^d$.  Depending on the relation between the dimension $d$
and the number of data points $n$, two main regimes have emerged: low- and
high-dimensional. The low-dimensional regime corresponds to $d = o(\log n)$;
 hence algorithms can afford to be \emph{exponential} in the dimension. In
the low-dimensional regime, efficient ANN algorithms are known for
{\em any metric space} \cite{Cl3, KR02, KL04, BKL06}.  In this paper, we focus on
the \emph{high-dimensional} regime, when $\omega(\log n) \leq d \leq
n^{o(1)}$, which is relevant for many applications.

The best-studied metrics are the Hamming ($\ell_1$) and the Euclidean
($\ell_2$) distances. There are good reasons for this: $\ell_1$ and
$\ell_2$ are very common in applications and admit very
efficient algorithms based on \emph{hashing}, in particular,
Locality-Sensitive Hashing (LSH)~\cite{IM98, AI06} and its
data-dependent versions~\cite{AINR14,AR15}.  Hashing-based algorithms
for ANN over $\ell_1$/$\ell_2$ have now been the subject of a two-decade-long line of
work, leading to a very good understanding of
algorithms and their limitations. All such algorithms for
$c$-approximate ANN obtain space $n^{1 + \rho_u + o(1)}$ and query
time $n^{\rho_q + o(1)}$ for some exponents $\rho_u$ and $\rho_q<1$
dependent on $c$; e.g., the most recent paper~\cite{ALRW16a}
gives \emph{tight} time--space trade-offs for every approximation
factor $c > 1$.\footnote{The exact dependence, for $\ell_2$, is that
one can achieve any $\rho_u, \rho_q \geq 0$ satisfying
$c^2 \sqrt{\rho_q} + (c^2 - 1) \sqrt{\rho_u} = \sqrt{2c^2 - 1}$.  }
We point the reader to \cite{HIM12}
and \cite{ALRW16a}, which summarize the state of affairs of the
high-dimensional ANN over $\ell_1/\ell_2$.  A practical perspective is
presented in the surveys~\cite{WSSJ14,WLKC15}.

Beyond $\ell_1$ and $\ell_2$, the landscape of ANN is much more
mysterious, despite having received significant attention.
In 1998, \cite{I01} showed
an efficient data structure for $\ell_\infty$ for $c=O(\log\log d)$
approximation. There are a few extensions of this result to other metrics, some of which proceed via \emph{embedding} a metric into $\ell_{\infty}$ (see
Section~\ref{sec_prior_work}). However, we are still very far from
having a general recipe for ANN data structures for \emph{general}
metrics with a non-trivial approximation; this is in stark contrast
with the success of the low-dimensional regime.
This state of affairs motivates the following broad question.

\begin{problem}
  \label{general_prob}
  For a given approximation $c > 1$, which metric spaces allow efficient ANN algorithms?
\end{problem}

An algorithm for general metrics is highly desirable both in theory and in practice.
From the theoretical perspective, we are interested in a
common theory of
ANN algorithms for a wide class of distances. Such a theory
would yield data structures (or impossibility results) for a~variety of important distance measures
for which we still do not know efficient ANN algorithms (e.g., matrix norms, the Earth Mover's Distance (EMD),
the edit distance, etc.). Perhaps even more tantalizing is understanding what
exactly makes some distances harder than others, and how to
quantify that hardness.  From the practical perspective, it is also desirable
to have a generic algorithm: one that either
uses the underlying distance measure as a black box, or
provides a ``knob'' to easily specialize
to any desired distance. In practice, one must oftentimes tune the distance to
the specifics of the application, and hence  algorithms that
allow such tuning without major re-implementations are preferred.

In this paper, we focus on the following important case of
Problem~\ref{general_prob}.

\begin{problem}
\label{general_prob_norms}
Solve Problem~\ref{general_prob} for high-dimensional \emph{normed} spaces.
\end{problem}

Norms are important for two reasons. First, most metric spaces arising
in applications are actually norms (e.g., the Earth-Mover
Distance \cite{NS07}). Second, norms are geometrically nicer than
general metrics, so there is hope for a coherent theory (e.g.,
for the problems of \emph{sketching} and {\em streaming} norms, see
the generic results of~\cite{AKR15, BBCKY15}).  Using
embeddings into $\ell_2$~\cite{John48,B97a}, one can solve ANN for \emph{any norm}
with approximation $O\Bigl(\sqrt{d / \eps}\Bigr)$, space $n^{1
+ \eps}$, and query time $n^{\eps}$, where $0 < \eps < 1/2$ is a
constant; however, no better results are known in general.

\subsection{Our main result} 

In this paper we nearly settle Problem~\ref{general_prob_norms}
for \emph{symmetric} norms, i.e., norms that are invariant under all
permutations and changes of signs of the coordinates of a vector.  We
show the following general result:

\begin{theorem}
  \label{main_thm}
  For every $n$,
  $d = n^{o\left(1\right)}$,
  and every $d$-dimensional symmetric norm $\|\cdot\|$,
  there exists a data structure for ANN over $\|\cdot\|$
  for $n$-point datasets
  with approximation $(\log \log n)^{O(1)}$
  space $n^{1 + o(1)}$, and query time $n^{o(1)}$.
\end{theorem}

We note that the techniques behind Theorem~\ref{main_thm} cannot be
extended to {\em general} norms; see details in
Section~\ref{sec_lower_bound}.

\subsection{Why symmetric norms?}
\label{ex_symmetric}
The class of symmetric norms is, in some sense, a sweet spot.  On the
one hand, symmetric norms are mathematically nice and, as we show,
allow for a clean characterization that leads to an efficient ANN data
structure (see the proof overview from
Section~\ref{sec:proofOverview}).  On the other hand, symmetric norms
vastly generalize $\ell_p$ distances and enable many new interesting
examples, some of which arise in applications.  We first consider
the following two examples of symmetric norms, which are crucial for
the subsequent discussion.
% (see Section~\ref{sec:moreSymmetric} for many more examples).

The first important example is the \emph{top-$k$ norm:} the sum of $k$
largest absolute values of the coordinates of a vector; $k=1$
corresponds to $\ell_{\infty}$, while $k=d$ corresponds to $\ell_1$.
Another rich set of examples is that of \emph{Orlicz norms:} for any
non-zero convex function $G: \Rbb_+ \to \Rbb_+$ such that $G(0) = 0$,
we define the \emph{unit ball} of a norm $\|\cdot\|_G$ to be:
$$
\Bigl\{x \in \Rbb^d \,\Bigm|\, \sum_{i=1}^d G\bigl(|x_i|\bigr) \leq 1\Bigr\}.
$$
Clearly, for $1 \leq p < \infty$ the $\ell_p$ norm is Orlicz via $G(t) = t^p$.

In statistics and machine learning, Orlicz norms are known as
\emph{$M$-estimators} (for the case of convex losses)~\cite{CW15}. A specific example
is the \emph{Huber loss}. Even though \emph{non-convex} losses do not correspond to norms,
our algorithm still can handle them (see Section~\ref{sec:orlicz}).

\vspace{1ex}\noindent {\bf Other examples of symmetric norms}
used in applications include:
\begin{itemize}[leftmargin=0in,noitemsep]
\item \emph{$k$-support norm}~\cite{AFS12} used for the sparse regression problem; its unit ball is
  the convex hull of $\{x \mid \mbox{$x$ is $k$-sparse}, \|x\|_{2} \leq 1\}$,
\item \emph{box-$\Theta$ norm}~\cite{MPS14} (again, used for sparse regression), defined for $0 < a < b \leq c$ and
  $\Theta = \{\theta \in [a, b]^d \mid \|\theta\|_{1} \leq c\}$ as
$
\|x\| = \min_{\theta \in \Theta} \left(\sum_{i=1}^d \frac{x_i^2}{\theta_i}\right)^{1/2},
$
and its dual;
\item \emph{$K$-functional}~\cite{DM93} used to show tight tail bounds, defined for $t > 0$ as
$
\|x\| = \min\Bigl\{\|x_1\|_{1} + t \cdot \|x_2\|_{2} \,\Bigm|\, x_1 + x_2 = x\Bigr\},
$
\item \emph{$\|\cdot\|_{1,2,s}$ norms}~\cite{KW16a} used for dimension reduction, defined as
$
\|x\| = \left(\sum_{i} \|x_{S_i}\|_{1}^2\right)^{1/2},
$
where $S_1$ is the set of  $s$ largest absolute values of coordinates of $x$, $S_2$ is the set of next $s$ largest coordinates, etc.
\end{itemize}

Finally, we show two simple ways to construct many interesting
examples of symmetric norms. Let $0 = a_0 \leq a_1 \leq
a_2 \leq \ldots \leq a_d$ be a non-decreasing
sub-additive\footnote{For every $n, m$, one has $a_{n+m} \leq a_n +
a_m$.} sequence.  We can define two norms associated with
it~\cite{BS88}: a \emph{minimal} norm is defined as
$$
\|x\| = \max_{1 \leq k \leq d} a_k \cdot \left(\mbox{average of the largest $k$ absolute values of the coordinates of $x$}\right),
$$
and a \emph{maximal} norm is equal to
$$
\|x\| = \sum_{k=1}^d \left(a_k - a_{k-1}\right)\cdot \left(\mbox{$k$-th largest absolute value of a coordinate of $x$}\right).
$$
The minimal norm is the \emph{smallest} norm such that for every $k$ one has:
$$
\Bigl\|(\underbrace{1, 1, \ldots, 1}_{k}, 0, 0, \ldots, 0)\Bigr\| = a_k.
$$
Similarly, the maximal norm is the \emph{largest} such norm.
Minimal norms will provide hard examples of symmetric norms
that preclude some simple(r) approaches to ANN (see Section~\ref{logd_lower}).
We also note that the dual (with respect to the standard dot
product) of any symmetric norm is symmetric as well.

\subsection{Prior work: ANN for norms beyond $\ell_1$ and $\ell_2$}
\label{sec_prior_work}

For norms beyond $\ell_1$ and $\ell_2$, the cornerstone result in ANN
 is a data structure for $\ell_{\infty}$ due to
Indyk~\cite{I01}. For every $\eps > 0$, the data structure
achieves space $n^{1 + \eps}$, query time $n^{o(1)}$, and
approximation $O_{\eps}(\log \log d)$. This is a
\emph{doubly-exponential} improvement over embeddings of
$\ell_{\infty}$ into $\ell_1 / \ell_2$ which require distortion
$\Omega\bigl(\sqrt{d}\bigr)$.

It is well-known~\cite{W91} that \emph{any} $d$-dimensional
normed space embeds into $\ell_{\infty}$ with distortion $(1 + \eps)$, 
which raises the question: can we combine this embedding with the result from
\cite{I01} to solve ANN for
any norm? It turns out that the answer is negative:
accommodating a norm of
interest may require embedding into a very high-dimensional
$\ell_{\infty}$. In the worst case, we need
$2^{O_{\eps}(d)}$ dimensions, and this bound is known
to be tight~\cite{B97a}, even for spaces as simple as $\ell_2$. Even
though this approach would give a non-trivial approximation of $O(\log \log 2^{O(d)}) =
O(\log d)$,
the resulting data structure has query time which is \emph{exponential} in~$d$; thus, this approach is interesting only for the low-dimensional
regime $d = o(\log n)$.

The result of~\cite{I01} has been extended as follows.
In~\cite{I02,I04,AIK09,A09} it was shown
how to build data structures for ANN over
arbitrary \emph{$\ell_p$-products} of metrics given that there exists
an ANN data structure for every factor.  Recall that the
$\ell_p$-product of metric spaces $M_1$, $M_2$, \ldots, $M_k$ is a
metric space with the ground set $M_1 \times M_2 \times \ldots \times
M_k$ and the following distance function:
$$
d\bigl((x_1, x_2, \ldots, x_k), (y_1, y_2, \ldots, y_k)\bigr) =
\Bigl\|\bigl(d_{M_1}(x_1, y_1), d_{M_2}(x_2, y_2), \ldots, d_{M_k}(x_k, y_k)\bigr)\Bigr\|_{p}.
$$
In a nutshell, if we can build efficient ANN data structures for every $M_i$ with approximation $c$,
there exist an efficient data structure for ANN over the product space with approximation $O(c \cdot \log \log n)$. Note that the above also implies ANN for the standard $\ell_p$, though for this
case
a better approximation $O(\log\log d)$ is possible via randomized embeddings into $\ell_{\infty}$~\cite{A09}.

For small values of $p$, one can also get $c=2^{O(p)}$~\cite{NR-ellp,
BG15} using different techniques.

\subsection{Overview of the proof of Theorem~\ref{main_thm}} 
\label{sec:proofOverview}
We prove Theorem~\ref{main_thm} in three steps.
\begin{itemize}[leftmargin=0in]
\item
  First, we build a data structure for $d$-dimensional top-$k$ norms.
We proceed by constructing a~\emph{randomized} embedding into $d$-dimensional $\ell_{\infty}$
with constant distortion, and then invoke the data structure for ANN over $\ell_{\infty}$
from~\cite{I01}.

Our embedding is a refinement of the technique of \emph{max-$p$-stable
distributions} used in~\cite{A09} to embed $\ell_p$ into
$\ell_{\infty}$. Surprisingly, the technique turns out to be very
general, and can handle top-$k$ norms as well an \emph{arbitrary}
Orlicz norm.

While this technique can handle even arbitrary \emph{symmetric norms} (see Appendix~\ref{sec:simple-sym}), there exist
symmetric norms, for which this approach leads to merely a
$\log^{\Omega(1)} d$-approximation, which is exponentially worse than
the bound we are aiming at (see Section~\ref{logd_lower}).

\item To bypass the above limitation and obtain the desired
  $(\log \log n)^{O(1)}$-approximation, we show the following
structural result: \emph{any} $d$-dimensional symmetric norm allows a constant-distortion (deterministic) embedding into
a low-dimensional \emph{iterated product} of top-$k$ norms. More
specifically, the host space $Y$ is an $\ell_{\infty}$-product of
$d^{O(1)}$ copies of the $\ell_1$-product of $X_1$,
$X_2$, \ldots, $X_d$, where $X_k$ is $\Rbb^d$ equipped with the
top-$k$ norm.

  The dimension of $Y$ is $d^{O(1)}$ which is significantly
  better than the bound $2^{\Omega(d)}$ necessary to embed symmetric norms (even $\ell_2$)
  into $\ell_{\infty}$. It is exactly this improvement over the na\"{\i}ve approach
  that allows us to handle any dimension $d = n^{o(1)}$
  as opposed to the trivial $o(\log n)$.

\item
  Finally, we use known results~\cite{I02, A09}, which allow us to construct
  a data structure for ANN over a product space if we have ANN data
  structures for the
  individual factors. Each such step incurs
  an additional $\log \log n$ factor in the resulting approximation.
  Since we have built a data structure for top-$k$ norms,
  and can embed a symmetric norm into an iterated product of top-$k$
  norms, we are done!
  
  Embeddings into iterated product spaces have been successfully used
  before for constructing data structures for ANN over Fr\'{e}chet
  distance~\cite{I02}, edit distance~\cite{I04}, and Ulam
  distance~\cite{AIK09}. Theorem~\ref{main_thm} gives yet another
  confirmation of the power of the technique.
\end{itemize}

\subsection{Optimality of Theorem~\ref{main_thm}}

There remains one aspects of Theorem~\ref{main_thm} that can potentially
be improved: the approximation factor $(\log \log n)^{O(1)}$.

One of the bottlenecks for our algorithm is the ANN data structure for $\ell_{\infty}$ from~\cite{I01},
which gives $O(\log \log d)$ approximation. This bound is known to be tight~\cite{ACP08, KP12}
for certain models of computation (in particular, for decision trees, which captures
the result of~\cite{I01}). Thus, going beyond
approximation $\Omega(\log \log d)$ in Theorem~\ref{main_thm} might be hard;
however, it remains entirely possible to
improve the approximation from $(\log \log n)^{O(1)}$ to $O(\log \log
d)$, which we leave as an open question.

\subsection{Lower bounds for general norms}
\label{sec_lower_bound}
The second step of the proof of Theorem~\ref{main_thm} (see Section~\ref{sec:proofOverview}) shows how to embed any $d$-dimensional
symmetric norm into a \emph{universal} normed space of dimension
$d^{O(1)}$ with a constant distortion. In contrast, we show
that for \emph{general} norms a similar universal construction is
impossible.  More formally, for a fixed $0 < \eps < 1/3$, suppose $U$
is a normed space such that for every $d$-dimensional normed space $X$
there exists a \emph{randomized} linear embedding of $X$ into $U$ with
distortion $O(d^{1/2 - \eps})$. Then,~$U$ must have dimension at least
$\exp\left(d^{\Omega_{\eps}(1)}\right)$. By John's theorem~\cite{John48},
$d$-dimensional $\ell_2$ is a~universal space for distortion
$\sqrt{d}$, so our lower bound is tight up to sub-polynomial factors. See
Section~\ref{sec:linearLB} for details.

To take this a step further, it would be highly desirable to prove
stronger hardness results for ANN over general norms. One approach would be to
show that such a norm $X$ has high \emph{robust expansion}, which is a
property used to deduce ANN lower
bounds~\cite{PTW10, ALRW16a}. There exist \emph{metrics} $M$ that have high
robust expansion, such as the shortest path metric of a~spectral expander (see
Appendix~\ref{apx:lbMetrics}). To obtain a hard norm, it suffices to embed such an
$N$-point metric $M$ into a $\log^{O(1)} N$-dimensional norm with a constant distortion.
The result of \cite{Mat96} shows that
there exist $N$-point metrics $M$ which {\em cannot} be embedded into any norm
of dimension $N^{o(1)}$. However, these metrics are not expanders,
and for expanders such a dimension reduction procedure might be possible. \footnote{In a recent work, Naor \cite{N17} showed that this approach is impossible. He shows that embedding an $N$-point spectral expander with constant distortion into any normed space requires $\poly(N)$ dimensions.}

\subsection{Other related work: dealing with general norms}

The recent result of~\cite{BBCKY15} completely characterizes
the \emph{streaming} complexity of any symmetric norm.  Even though
many symmetric norms (including $\ell_{\infty}$) are hard in
the streaming model, the state of affairs with ANN is arguably much
nicer. In particular, our results imply that all symmetric norms have
 \emph{highly efficient} ANN data structures.  We also point
out that streaming algorithms for the special case of \emph{Orlicz}
norms have been studied earlier~\cite{BO10}.

Another related work is~\cite{AKR15}, which shows that for
norms, the existence of good \emph{sketches} is equivalent to
\emph{uniform} embeddability into $\ell_2$. Sketches are known
to imply efficient ANN data structures, but since many symmetric norms
do not embed into $\ell_2$ uniformly, we conclude that ANN is provably easier
than sketching for a large class of norms.

Finally, we also mention the work of \cite{AV15}, who study ANN under
the class of high-dimensional distances which are {\em Bregman
divergences}. These results are somewhat disjoint since the Bregman
divergences are not norms.

\section{Preliminaries}

% !TEX root = main.tex

\subsection{Norms and products}

We denote non-negative real numbers by $\R_{+}$. For any subset $A \subseteq \R$, we let $\chi_{A} \colon \R \to \{0, 1\}$ be the indicator function of $A$.
Let $X$ be a normed space over $\R^d$. We denote $B_X$ the unit ball of $X$, and
$\|\cdot \|_{X}$ the norm of $X$. We denote $X^*$ the dual norm of
$X$ with respect to the standard dot product $\langle \cdot, \cdot \rangle$, i.e~$\|x\|_{X^*} =
\sup\{|\langle x, y\rangle|: y \in B_X\}$.
For a vector $x \in \Rbb^d$ we define $|x| = (|x_1|, |x_2|, \ldots, |x_d|)$
to be the vector of the absolute values of the coordinates of $x$. For a positive integer $d$ and $1 \leq p \leq \infty$, we denote $\ell_p^d$ the space $\Rbb^d$ equipped with the standard $\ell_p$ norm, which we denote by $\|\cdot\|_p$.

\begin{definition}
\label{def:ordered}
For any vector $x \in \R^d$, we let $x^* = P|x|$ be the vector obtained by applying the permutation matrix $P$ to $|x|$ so coordinates of $x^*$ are sorted in non-increasing absolute value. 
\end{definition}
                                           
\begin{definition}[Symmetric norm]
A norm $\|\cdot\|_{X} \colon \R^d \to \R$ is \emph{symmetric} if for every $x \in \R^d$, $\|x\|_X = \Bigl\||x|\Bigr\|_X = \|x^*\|_{X}$.
\end{definition}

See the introduction for examples of symmetric norms. We note once
again that the dual norm of a symmetric norm is also symmetric.

A natural way to combine norms is via {\em product spaces}, which we
will heavily exploit in this paper.

\begin{definition}[Product space]
  Let $1 \leq p \leq \infty$. Let $(X_1, d_{X_1})$, $(X_2, d_{X_2})$,
  \ldots, $(X_k, d_{X_k})$ be metric spaces.  We define the {\em
    $\ell_p$-product space}, denoted $\bigoplus_{\ell_p} X_i$, to be a
  metric space whose ground set is $X_1 \times X_2 \times \ldots
  \times X_k$, and the distance function is defined as follows: the
  distance between $(x_1, x_2, \ldots, x_k)$ and $(x_1', x_2', \ldots,
  x_k')$ is defined as the $\ell_p$~norm of the vector
  $\bigl(d_{X_1}(x_1, x_1'), d_{X_2}(x_1, x_2'), \ldots, d_{X_k}(x_k,
  x_k')\bigr)$.
\end{definition}

Next we define the top-$k$ norm:
\begin{definition}
For any $k \in [d]$, the \emph{top-$k$ norm}, $\| \cdot \|_{T(k)} \colon \R^d \to \R$, is the sum of the absolute values of the top $k$ coordinates. In other words, 
\[ \| x\|_{T(k)} = \sum_{i=1}^k |x_i^*|, \]
where $x^*$ is the vector obtained in Definition~\ref{def:ordered}.
\end{definition}

\begin{definition}
Given vectors $x, y \in \R^d$, we say $x$ \emph{weakly majorizes} $y$ if for all $k \in [d]$,
\[ \sum_{i=1}^k |x^*_i| \geq \sum_{i=1}^k |y^*_i|. \]
\end{definition}

\begin{lemma}[Theorem B.2 in \cite{MOA11}]
If $x, y \in \R^d$ where $x$ weakly majorizes $y$, then for any symmetric norm $\|\cdot\|_{X}$,
\[ \|x\|_{X} \geq \|y\|_{X}. \]
\end{lemma}

\begin{definition}
For $i \in [d]$, let $\xi^{(i)} \in \R^d$ be the vector
\[ \xi^{(i)} = (\underbrace{1, \dots, 1}_{i}, \underbrace{0, \dots, 0}_{d-i}) \]
consisting of exactly $i$ 1's, and $d-i$ 0's.
\end{definition}

\subsection{ANN for $\ell_{\infty}$ and $\ell_{\infty}$-products}

We will crucially use the following two powerful results of Indyk. The
first result is for the standard $d$-dimensional $\ell_\infty$ space.

\begin{theorem}[{\cite[Theorem~1]{I01}}]
\label{thm:l-infty-ds}
For any $\eps \in (0, 1/2)$, there exists a data structure for ANN for
$n$-points datasets in the $\ell_{\infty}^d$ space with approximation
$O\left(\frac{\log \log d}{\eps}\right)$, space $O(d\cdot n^{1
+ \eps})$, and query time $O(d \cdot \log n)$.
\end{theorem}

The second is a generalization of the above theorem, which applies to
an $\ell_{\infty}$-product of $k$ metrics $X_1,\ldots X_k$, and achieves
approximation $O(\log \log n)$. It only needs black-box ANN schemes
for each metric $X_i$.

\begin{theorem}[{\cite[Theorem~1]{I02}}]
\label{thm:maxProduct}
Let $X_1, X_2, \ldots, X_k$ be metric space, and let $c > 1$ be a real
number.  Suppose that for every $1 \leq i \leq k$ and every $n$ there
exists a data structure for ANN for $n$-point datasets from $X_i$ with
approximation $c$, space $S(n)\ge n$, query time $Q(n)$, and probability of
success $0.99$. Then, for
every $\eps>0$, there exists ANN under
$\bigoplus_{\ell_\infty}^k{\cal M}$ with:
\begin{itemize}
\item
$O(\eps^{-1}\log\log n)$ approximation, 
\item
$O(Q(n)\log n+dk\log n)$ query time, where $d$ is the time to compute
  distances in each $X_i$, and
\item $S(n)\cdot O(kn^{\eps})$ space/preprocessing.
\end{itemize}
\end{theorem}

Strictly speaking, we need to impose a technical condition on the ANN
for each $X_i$ --- that it reports the point with the smallest {\em
  priority} --- which is satisfied in all our scenarios;
see~\cite[Section 2]{I02} for details. Also, the original statement
of \cite{I02} gave a somewhat worse space bound. The better space
results simply from a better analysis of the algorithm, as was
observed in \cite{AIK09}; we include a proof in Appendix
\ref{apx:space}.

\section{ANN for Orlicz and top-$k$ norms}
\label{sec:orlicz}
% !TEX root = main.tex
Before showing a data structure for general symmetric norms, we give an algorithm for general Orlicz norms. We then show how to apply these ideas to top-$k$ norms. This restricted setting has a simple analysis and illustrates one of the main techniques used in the rest of the paper.
A similar approach was used in prior work to construct randomized
embeddings of $\ell_p$ norms into $\ell_\infty$, and solve the ANN
search problem; here we show that these techniques are in fact
applicable in much greater generality.  

\begin{lemma}
  \label{lem:orliczmap}
  Let $\|\cdot\|_G$ be an Orlicz norm. For every $D, \alpha > 1$ and every $\mu \in (0, 1/2)$
  there exists a randomized linear map $f \colon \Rbb^d \to \Rbb^d$ such that for every $x \in \Rbb^d$:
  \begin{itemize}
    \item if $\|x\|_G \leq 1$, then $\mathrm{Pr}_f\Bigl[\bigl\|f(x)\bigr\|_{\infty} \leq 1\Bigr] \geq \mu$;
    \item if $\|x\|_G > \alpha D$,
      then $\mathrm{Pr}_f\Bigl[\bigl\|f(x)\bigr\|_{\infty} > D\Bigr] \geq 1 - \mu^{\alpha}$.
  \end{itemize}
\end{lemma}

\begin{proof}
Let the distribution $\calD$ over $\R_+$ have the following CDF $F: \R_+ \to [0, 1]$:
$$
F(t) = \Pr_{u \sim \calD}[ u \leq t ] = 1 - \mu^{G(t)}.
  $$
Consider the following randomized linear map $f: \R^d \to \R^d$:
\[ \left( x_1, x_2, \dots, x_d \right) \mathop{\mapsto}^f \left( \frac{x_1}{u_1}, \frac{x_2}{u_2}, \dots, \frac{x_d}{u_d} \right)\]
where $u_1, \dots, u_d \sim \calD$ are i.i.d.\ samples from $\calD$.
Suppose that $\|x\|_G \leq 1$. Then, $\sum_{i=1}^d G(|x_i|) \leq 1$. This, in turn,
implies:
\[ \Pr_f\Bigl[ \|f(x)\|_{\infty} \leq 1\Bigr] = \prod_{i=1}^d \Pr_{u_i \sim \calD}\left[ \left|\frac{x_i}{u_i}\right| \leq 1 \right] = \prod_{i=1}^d \mu^{G(|x_i|)} = \mu^{\sum_{i=1}^d G(|x_i|)} \geq \mu. \]
Now suppose that $\|x\|_G > \alpha D$. This, together with the convexity of $G(\cdot)$, implies:
$$
\sum_{i=1}^d G\left(\frac{|x_i|}{D}\right) \geq (1-\alpha)G(0) + \alpha \cdot \sum_{i=1}^d G\left(\frac{|x_i|}{\alpha D}\right) \geq \alpha.
$$
Thus, we have:
\[ \Pr_f\Bigl[ \|f(x)\|_{\infty} \leq D\Bigr] = \prod_{i=1}^d \Pr_{u_i \sim \calD}\left[ \left|\frac{x_i}{u_i}\right| \leq D\right] = \prod_{i=1}^d \mu^{G\left(|x_i|/D\right)} =\mu^{\sum_{i=1}^d G(|x_i|/D)} \leq \mu^{\alpha}. \]
\end{proof}

\begin{theorem}
\label{thm:dsorlicz}
For every $d$-dimensional Orlicz norm $\|\cdot\|_G$ and every $\eps \in (0, 1/2)$, there exists a data structure for ANN over $\|\cdot\|_G$, which achieves approximation $O\left(\frac{\log \log d}{\eps^2}\right)$ using space $O\left(d n^{1 + \eps}\right)$ and query time $O\left(d n^{\eps}\right)$.
\end{theorem}

\begin{proof}
Let $P \subset \R^d$ be a dataset of $n$ points. Consider the data structure which does the following: 
\begin{enumerate}
\item For all $1 \leq i \leq n^{\eps}$, we independently apply the randomized linear map $f$ from Lemma~\ref{lem:orliczmap} with parameters $\mu = n^{-\eps}$, $D = O\left( \frac{\log \log d}{\eps}\right)$, and $\alpha = \frac{2}{\eps}$. We define 
\[ P_i = \{ f_i(x) \mid x \in P \} \]
to be the image of the dataset under $f_i$, where $f_i$ is the $i$-th independent copy of $f$.
\item For each $1 \leq i \leq n^{\eps}$, we use Theorem~\ref{thm:l-infty-ds} to build a data structure for ANN over $\ell_{\infty}$ with approximation $D$ for dataset $P_i$. We refer to the $i$-th data structure as $T_i$.  
\end{enumerate}
Each $T_i$ occupies space $O(dn^{1 + \eps})$ and achieves approximation $D$ with query time $O(d\log n)$. To answer a query $q \in \R^d$, we query $T_i$ with $f_i(q)$ for each $i \in [n^{\eps}]$. Let $x_i$ be the point returned by $T_i$, and let $p_i \in P$ be the pre-image of $x_i$ under $f_i$, so that $f_i(p_i) = x_i$. If for some $T_i$, the point returned satisfies $\|p_i - q\|_{G} \leq \alpha D$, then we return $p_i$.
\begin{itemize}
\item If there exists some $p \in P$ with $\|p-q\|_{G} \leq 1$, then by Lemma~\ref{lem:orliczmap}, with probability $1 - \left(1 - n^{-\eps} \right)^{n^{\eps}} \geq \frac{3}{5}$, some $f_i$ has $\|f_i(p-q)\|_{\infty} \leq 1$. Since $f_i$ is linear, $\|f_i(p) - f_i(q)\|_{\infty} \leq 1$ as well. 
\item Let $i \in [n^{\eps}]$ be an index where some $p \in P$ with $\|p - q\|_{G} \leq 1$ has $\|f_i(p) - f_i(q)\|_{\infty} \leq 1$. Every other $p' \in P$ with $\|p' - q\|_{G} \geq \alpha D$ satisfies 
\[ \Pr\Bigl[ \|f_i(p') - f_i(q)\|_{\infty} \leq D\Bigr] \leq \frac{1}{n^2}. \]
A union bound over at most $n$ points with distance greater than
$\alpha D$ to $q$ shows that except with probability at most $\frac{1}{n}$, $T_i$ returns some $p_i \in P$ with $\|p_i - q\|_{G} \leq \alpha D$. Thus, the total probability of success of the data structure is at least $\frac{3}{5} - \frac{1}{n}$.
\end{itemize}
The total query time is $O\left( dn^{\eps} \cdot \log n\right)$ and the total space used is $O\left( dn^{1 + 2\eps}\right)$. This data structure achieves approximation $\alpha D = O\left( \frac{\log \log d}{\eps^2}\right)$. Decreasing $\eps$ by a constant factor, we get the desired
guarantees.
\end{proof}

\paragraph{Remark.} 
The construction of the randomized embedding in
Lemma~\ref{lem:orliczmap} and the data structure from
Theorem \ref{thm:dsorlicz} work in a somewhat more general setting,
rather than just for Orlicz norms. For a fixed norm $\|\cdot \|$, we
can build a randomized map $f \colon \R^d \to \R^d$ with the
guarantees of Lemma~\ref{lem:orliczmap} if there exists a
non-decreasing $G \colon \R_+ \to \R_+$ where $G(0) = 0$,
$G(t) \to \infty$ as $t \to \infty$, and for every $x \in \R^d$:
\begin{itemize}
\item if $\|x\| \leq 1$, then $\sum_{i=1}^d G(|x_i|) \leq 1$, and
\item if $\|x\| \geq \alpha D$, then $\sum_{i=1}^d G\left(\frac{|x_i|}{D}\right) \geq \alpha$. 
\end{itemize}
The data structure itself just requires the existence of a randomized
linear map satisfying the conditions of Lemma~\ref{lem:orliczmap}. 

We now describe how to obtain a data structure for ANN for any top-$k$ norm.

\begin{lemma}
\label{lem:kyfanmap}
Fix any $k \in [d]$. For every $D, \alpha > 1$ and every $\mu \in (0, 1/2)$, there exists a randomized linear map $f \colon \R^d \to \R^d$ such that for every $x \in \R^d$:
\begin{itemize}
\item if $\|x\|_{T(k)} \leq 1$, then $\Pr_f\Bigl[ \|f(x)\|_{\infty} \leq 1 \Bigr] \geq \mu$;
\item if $\|x\|_{T(k)} > \alpha D$, then $\Pr_{f} \Bigl[ \|f(x)\|_{\infty} > D\Bigr] \geq 1 - \mu^{\alpha - 1}$.
\end{itemize}
\end{lemma}

\begin{proof}
We define $G \colon \R_{+} \to \R_{+}$ where for every $x \in \R^d$, 
\[ G(t) = t \cdot \chi_{[\frac{1}{k}, \infty)}(t) \]
If $\|x\|_{T(k)} \leq 1$, there are at most $k$ coordinates where
$|x_i| \geq \frac{1}{k}$. Therefore, $\sum_{i=1}^d G(|x_i|) \leq
\|x\|_{T(k)} \leq 1$. If $\|x\|_{T(k)} \geq \alpha D$, then
$\sum_{i=1}^k |x_i^*| \geq \alpha D$. Therefore, $\sum_{i=1}^d
G\left(\frac{|x_i^*|}{D}\right) \ge \sum_{i=1}^k G\left(\frac{|x_i^*|}{D}\right) \geq \alpha - 1$. The proof now follows in the same way as Lemma~\ref{lem:orliczmap}.
\end{proof}

Lemma~\ref{lem:kyfanmap} gives us a data structure for any top-$k$ norm with approximation $O(\log\log d)$ applying Theorem~\ref{thm:dsorlicz}.

One could imagine using a similar argument to design an algorithm for \emph{general} symmetric norms. This idea indeed works and yields an algorithm with approximation $\widetilde{O}(\log d)$ for a general symmetric norm (see Appendix~\ref{sec:simple-sym} for a detailed analysis of this approach). However, we show this strategy cannot achieve an approximation better than $\Omega(\sqrt{\log d})$ (see the end of the same Appendix~\ref{sec:simple-sym}).

\section{Embedding symmetric norms into product spaces}
\label{sec:prod-space}
% !TEX root = main.tex

In this section, we construct an embedding of general symmetric norms
into product spaces of top-$k$ norms. To state the main result of this section, we need
the following definition.

\begin{definition}
For any $c_1, \dots, c_d \geq 0$, let $\bigoplus_{\ell_1}^{d} T^{(c)} \subset \R^{d^2}$ be the space given by the seminorm $\|\cdot\|_{T, 1}^{(c)} \cdot \R^{d^2} \to \R$ where for $x = (x_1, \dots, x_d) \in \R^{d^2}$ and $x_1, \dots, x_d \in \R^d$\emph{:}
\[ \|x\|_{T, 1}^{(c)} = \sum_{k=1}^d c_k\|x_k\|_{T(k)}. \]
\end{definition}

\begin{theorem}[Embedding into a product space]
\label{thm:symnorm}
For any constant $\delta \in (0, 1/2)$, any symmetric norm
$\|\cdot \|_{X} \colon \R^d \to \R$ can be embedded linearly with distortion
$1+\delta$ into $\bigoplus_{\ell_{\infty}}^t \bigoplus_{\ell_1}^d
T^{(c)}$ where $t
= d^{O(\log(1/\delta)\delta^{-1})}$.
In particular, there exists $c \in \R_+^{t\times d}$ such that for
every $x \in \R^d$,  
\begin{equation}
\label{eq:embedding} (1 - \delta)\|x\|_{X} \leq \max_{i \in [t]} \left(\sum_{k=1}^d c_{i, k} \|x\|_{T(k)} \right) \leq (1 + \delta)\|x\|_{X}.
\end{equation}
\end{theorem}
The vectors in $\bigoplus_{\ell_{\infty}}^t \bigoplus_{\ell_1}^d
T^{(c)} \subset \R^{td^2}$ can be broken up into $td$ blocks of $d$
coordinates each. The embedding referenced above will simply map $x
\in \R^d$ into $\R^{td^2}$ by making each of the $td$ many blocks
equal to a copy of $x$. The non-trivial part of the above theorem is
setting the constants $c_{i, k}$ for $i \in [t]$ and $k \in [d]$ so
(\ref{eq:embedding}) holds. %Theorem~\ref{thm:symnorm} can be stated
%more generally in terms of a function of the size of a certain
%$\gamma$-net. We need the following definitions. 
Before going on to give the proof of Theorem~\ref{thm:symnorm}, we establish
some definitions and propositions which will be used in the proof. For the subsequent sections, let $\beta \in (1, 2)$ be considered a constant close to $1$. 

\begin{definition}[Levels and Level Vectors]
\label{def:levels}
For any fixed vector $x \in \R^d$ and any $k \in \Z$, we define \emph{level $k$ with respect to $x$} as $B_k(x) = \{ i \in [d] \mid \beta^{-k-1} < |x_i| \leq \beta^{-k} \}$. Additionally, we let $b_k(x) = |B_k(x)|$ be the \emph{size of level $k$ with respect to $x$}. The \emph{level vector} of $x$, $V(x) \in \R^d$ is given by
\[ V(x) = (\underbrace{\beta^{k}, \dots, \beta^{k}}_{\text{$b_{-k}(x)$ times}}, \underbrace{\beta^{k-1}, \dots, \beta^{k-1}}_{\text{$b_{-k+1}(x)$ times}}, \dots, \underbrace{\beta^{-k}, \dots, \beta^{-k}}_{\text{$b_k(x)$ times}}, 0, \dots 0) \]
where $k$ is some integer such that all non-zero coordinates lie in some level between $-k$ and $k$. 
We say the $i$-th \emph{level vector} $V_i(x) \in \R^d$ is given by
\[ V_i(x) = (\underbrace{\beta^{-i}, \dots, \beta^{-i}}_{\text{$b_i(x)$ times}}, 0, \dots, 0). \]
\end{definition}

The notation used for level vectors appears in \cite{BBCKY15}; however, we refer to level $k$ as the coordinates of $x$ lying in $(\beta^{-k-1}, \beta^{-k}]$; whereas \cite{BBCKY15} refers to level $k$ as the coordinates of $x$ lying in $[\beta^{k-1}, \beta^k)$. 

\begin{definition}
\label{def:cut-small}
Fix some $\tau > 0$. For any vector $x \in \R^d$, let $C(x) \in \R^d$ be the vector where each $i \in [d]$ sets
\[ C(x)_i = \left\{ \begin{array}{cc} x_i & |x_i| \geq \tau \\
						    0 & |x_i| < \tau \end{array} \right. . \]
\end{definition}

\begin{proposition}[Proposition 3.4 in \cite{BBCKY15}]
\label{prop:levels}
Let $\|\cdot\|_{X}$ be any symmetric norm and $x \in \R^d$ be any vector. Then
\[ \frac{1}{\beta} \|V(x)\|_{X} \leq \|x\|_X \leq \|V(x)\|_{X}. \]
\end{proposition}

\begin{proposition}
\label{prop:cut-small-level}
Let $\|\cdot\|_{X}$ be any symmetric norm. For any vector $x \in \R^d$, 
\[ \|x\|_{X} - \tau d \leq \|C(x)\|_{X} \leq \|x\|_{X}. \]
\end{proposition}

\begin{proof}
Note that $x$ weakly majorizes $C(x)$, so $\|C(x)\|_{X} \leq \|x\|_{X}$. For the other direction, let $v = x - C(x)$. Then $v$ has all coordinates with absolute value at most $\tau$, so $\tau d \xi^{(1)}$ weakly majorizes $v$. Therefore,
\[ \|x\|_{X} \leq \|C(x)\|_{X} + \|v\|_{X} \leq \|C(x)\|_{X} + \tau d. \]
\end{proof}

Intuitively, the above two propositions say that up to multiplicative loss $\beta$ and additive loss $\tau d$ in the norm of the vector, we may assume that all coordinates are exactly $\beta^j$ for $j \geq \log_{\beta} (\tau)$. Thus, if $x \in \R^d$, then
\[ \|x\|_{X} - \tau d \leq \|V(C(x))\|_{X} \leq \beta\|x\|_{X}. \]
If additionally, we let $\tau = \frac{\beta}{d^2}$, so when $\|x\|_{X} \leq 1$ there are at most $2\log_{\beta} d$ non-empty levels in $V(C(x))$.

\begin{definition}[Rounded counts vector]
Fix any level vector $x \in \R^d$. The \emph{rounded counts vector} of $x$, $W(x) \in \R^d$ is given by $y$ where the $y \in \R^d$ is constructed using the following procedure:
\begin{algorithmic}[1]
\State Initialize $y = (0, \dots, 0) \in \R^d$ and $c = d$.
\For{$k=-\infty, \ldots, 2\log_{\beta}(d)-1$}
\If{$b_k(x) \geq 0$} 
\State Let $j \in \Z_+$ be the integer where $\beta^{j-1} < b_k(x) \leq \beta^{j}$. 
\If{$c \geq \lfloor \beta^j \rfloor$}
\State Set the first $\lfloor \beta^j \rfloor$ zero-coordinates of $y$ with $\beta^{-k}$. Update $c \leftarrow c - \lfloor \beta^{-k} \rfloor$.
\EndIf
\EndIf
\EndFor
\State Return $y$
\end{algorithmic}
\end{definition}

Intuitively, $W(x)$ represents the level vector of $x$ where we ignore coordinates smaller than $\frac{\beta}{d^2}$, and additionally, we round the counts of coordinates to powers of $\beta$. 
\begin{lemma}\label{lem:roundedbound}
For every vector $x \in \R^d$ and any symmetric norm $\|\cdot\|_{X}$,
\[ \|x\|_{X} - \tau d \leq \|W(V(C(x)))\|_{X} \leq \beta^2 \|x\|_{X}. \]
\end{lemma}

\begin{proof}
The bound $\|x\|_{X} - \tau d \leq \|W(V(C(x)))\|_{X}$ follows by combining Proposition~\ref{prop:levels} and Proposition~\ref{prop:cut-small-level}, as well as the monotonicity of norms. The bound $\|W(V(C(x)))\|_{X} \leq \beta^2 \|x\|_{X}$ follows from Proposition~\ref{prop:levels}, Proposition~\ref{prop:cut-small-level}, as well as Lemma 3.5 from \cite{BBCKY15}.
\end{proof}

In order to simplify notation, we let $R \colon \R^d \to \R^d$ given by $R(x) = W(V(C(x)))$. 

\def\Eps{\mathcal{E}}
\def\calR{\mathcal{R}}

\begin{definition}
Let the set $\calL \subset \R^d_+$ be given by
\[ \calL = \{ y \in \R^d_+ \mid y_1 \geq \dots y_d \geq 0 \}. \]
Additionally, for an arbitrary symmetric norm $\|\cdot\|_{X}$ with dual norm $\|\cdot\|_{X^*}$, we let the set $\calR \subset \calL$ be given by
\[ \calR = \{ R(y) \in \R^d_+ \mid y \in \calL \cap B_{X^*} \}. \]
\end{definition}

\begin{definition}
Fix a vector $y \in \calL \setminus \{ 0 \}$ ($y$ has non-negative, non-increasing coordinates). Let the \emph{maximal
seminorm} with respect to $y$, $\|\cdot\|_{y} \colon \R^d \to \R$ be the
seminorm where for every $x \in \R^d$,
\[ \|x\|_{y} = \langle |x^*|, y\rangle. \]
\end{definition}

We first show there exists some setting of $c \in \R^d$ such that we may compute $\|x\|_{y}$ as $\oplus_{\ell_1}^d T^{(c)}$. 

\begin{lemma}
\label{lem:hehe}
For every vector $y \in \calL \setminus \{0 \}$, there exists
$c_1, \dots, c_d \geq 0$ where for all $x \in \R^d$,
\[ \|x\|_{y} = \|x\|_{T, 1}^{(c)}. \]
\end{lemma}

\begin{proof}
For $k \in [d]$, we let $c_k = y_k - y_{k+1}$, where $y_{d+1} = 0$. 
\[ \langle |x^*|, y\rangle = \sum_{i=1}^d |x^*_i| y_i = \sum_{i=1}^d |x^*_i| \left(\sum_{k=i}^d c_k \right) = \sum_{k=1}^d c_k \left(\sum_{i=1}^k |x^*_i| \right) = \sum_{k=1}^d c_k \|x\|_{T(k)} \]
\end{proof}

Given Lemma~\ref{lem:hehe}, it suffices to show that for an arbitrary symmetric norm $\|\cdot\|_{X}$, we may compute $\|x\|_{X}$ (with some distortion) as a maximum over many maximal seminorms. In the following lemma, we show that taking the maximum over maximal norms from $\calR$ suffices, but gives sub-optimal parameters. We then improve the parameters to prove Theorem~\ref{thm:symnorm}.

\begin{lemma}\label{lem:overR}
Let $\|\cdot\|_{X}$ be an arbitrary symmetric norm and let $\|\cdot\|_{X^{*}}$ be its dual norm. 
Then for any $\|x\|_{X} \leq 1$,
\[ \|x\|_{X} - \tau d \leq \max_{y \in \calR} \|x\|_{y} \leq \beta^2 \|x\|_{X}. \]
\end{lemma}

\begin{proof}
Without loss of generality, we rescale the norm so that $\|e_1\|_{X^*} = 1$, where $e_1$ is the first
standard basis vector. Consider any $x \in \R^d$ with $\|x\|_{X} \leq 1$. Then since $\|\cdot\|_{X}$ is symmetric, we may assume without loss of generality that all coordinates of $x$ are non-negative and in non-increasing order. Thus for each $y \in \calL \cap \{ 0 \}$, we have $\|x\|_{y} = \langle x, y \rangle$. 

The lower bound simply follows from the fact that $R(z)$, other than coordinates less than $\tau$, is monotonically above $z$, and all coordinates in $x$ are non-negative. More specifically,
\[ \|x\|_{X} = \sup_{z \in \calL \cap B_{X^*}} \langle x, z\rangle \leq \sup_{z \in \calL \cap B_{X^{*}}} \langle x, R(z) \rangle + \tau d = \max_{y \in \calR} \ \langle x, y \rangle + \tau d,\]
where $\tau d$ comes from the fact that because $\|x\|_{X} \leq 1$, every coordinate of $x$ is at most $1$. 
On the other hand, we have
\begin{align*}
\max_{y \in \calR}\ \langle x, y \rangle = \beta^2 \max_{y \in \calR} \ \langle x, \frac{y}{\beta^2} \rangle \leq \beta^2 \sup_{z \in B_{X^*}} \langle x, z \rangle = \beta^2 \|x\|_{X},
\end{align*}
where we used the fact that $\|\frac{y}{\beta^2}\|_{X^*} \leq 1$ by Lemma~\ref{lem:roundedbound}.
\end{proof}

Given Lemma~\ref{lem:overR}, it follows that we may linearly embed $X$ into $\oplus_{\ell_{\infty}}^t \oplus_{\ell_1}^d T^{(c)}$ where $t = |\calR|$, with distortion $\frac{\beta^2}{1 - \tau d} \leq \beta^3$ (where we used the fact $\tau = \frac{\beta}{d}$ and that $1 + \beta / d \leq \beta$ for a large enough $d$). The embedding follows by copying the vector $x$ into the $t$ spaces $\oplus_{\ell_1}^d T^{(c)}$ corresponding to each vector $y \in \calR$ given in Lemma~\ref{lem:hehe}. The one caveat is that this embedding requires $t$ copies of $\oplus_{\ell_1}^d T^{(c)}$, and $t$ is as large as $\left(\log_{\beta} d + 1 \right)^{2\log_{\beta} d} = d^{O(\log \log d)}$. This is because there are at most $2 \log_{\beta} d$ many levels, and each contains has number of coordinates being  some value in $\{ \beta^{i} \}_{i=0}^{\log_{\beta} d}$. Thus, our algorithm becomes inefficient once $d \geq 2^{\omega\left(\frac{\log n}{\log \log n}\right)}$. 

In order to avoid this problem, we will make the embedding more efficient by showing that we do not need all of $\calR$, but rather a fine net of $\calR$. In addition, our net will be of polynomial size in the dimension, which gives an efficient algorithm for all $\omega(\log n) \leq d \leq n^{o(1)}$. We first show that it suffices to consider fine nets of $\calR$, and then build a fine net of $\calR$ of size $\poly(d)$. 

\begin{lemma}
\label{lem:hehehe}
Fix an $\gamma \in (0, 1/2)$. Let $\|\cdot\|_{X}$ be an arbitrary
symmetric norm and $\|\cdot\|_{X^{*}}$ be its dual norm. If $N$ is a
$\gamma$-net of $\calR$ with respect to distance given by $\|\cdot\|_{X^*}$, then
\[ (1-\gamma - \tau d)\|x\|_{X} \leq \max_{y \in N} \|x\|_{y} \leq (\beta^2 + \gamma) \|x\|_{X} \]
\end{lemma}

\begin{proof}
Since the embedding we build is linear, it suffices to show that every vector $x \in \R^d$ with $\|x\|_{X} = 1$ has 
\[ 1 - \gamma - \tau d \leq \max_{y \in N} \|x\|_y \leq (\beta^2 + \gamma). \]
Consider a fixed vector $x \in \R^d$ with $\|x\|_{X}=1$. Additionally, we may assume the coordinates of $x$ are non-negative and in non-increasing order. We simply follow the computation:
\begin{align*}
\|x\|_{X} &= \||x^*|\|_{X} \leq \max_{y \in \calR}\ \langle |x^*|, y \rangle + \tau d = \max_{y \in N} \left( \langle |x^*|, y\rangle + \langle |x^*|, v \rangle \right) + \tau d \leq \max_{y \in N} \langle |x^*|, y \rangle + \gamma \|x\|_{X} + \tau d,
\end{align*}
where we used Lemma~\ref{lem:overR} and the fact that $\|v\|_{X^*} \leq \gamma$ in a $\gamma$-net of $\calR$ with respect to the distance given by $\|\cdot\|_{X^*}$. On the other hand,
\begin{align*}
\max_{y \in N} \|x\|_y &= \max_{y \in \calR} \left( \langle |x^*|, y \rangle + \langle |x^*|, v \rangle \right) \leq \max_{y \in \calR} \|x\|_{y} + \gamma \|x\|_{X} \leq (\beta^2 + \gamma) \|x\|_{X},
\end{align*}
where again, we used Lemma~\ref{lem:overR} and the fact that $\|v\|_{X^*} \leq \gamma$. 
\end{proof}

Finally, we conclude the theorem by providing a $\gamma$-net for $\calR$ of size $d^{O(\log(1/\gamma)\gamma^{-1})}$.

\begin{lemma}
\label{lem:net}
Fix any symmetric space $X$ with dual $X^*$. There exists an $8(\beta - 1)$-net of size $d^{O(\log(1/(\beta - 1)) / \log \beta)}$ for $\calR$ with respect to distances given by $\|\cdot\|_{X^*}$.
\end{lemma}

We defer the proof of Lemma~\ref{lem:net} to the next section. 
The proof of Theorem~\ref{thm:symnorm} follows by combining
Lemma~\ref{lem:hehe}, Lemma~\ref{lem:hehehe}, and Lemma~\ref{lem:net}. In particular, given a $\frac{\beta-1}{8}$-net of $\calR$, we get an embedding with distortion at most $(\beta^2 + 8(\beta - 1)) (1 + (8(\beta - 1) + \tau d)^2)$ from Lemma~\ref{lem:hehehe}. We let $\tau = \frac{\beta}{d^2}$ and $\beta = \sqrt{1 + \delta/100}$ to get the desired linear embedding with distortion $1 + \delta$. We now proceed to proving Lemma~\ref{lem:net}, which gives the desired upper bound on the size of the net.

\subsection{Proof of Lemma~\ref{lem:net}: bounding the net size}

We now give an upper bound on the size of a fine net of
$\calR$. We proceed by constructing a further simplification of $R(x)$. Intuitively we show that one can ignore the higher levels if there are fewer coordinates in the higher levels than some lower level.

\begin{lemma}\label{lem:sum}
Let $\|\cdot\|_X$ be a symmetric norm. Consider any nonnegative vector $x \in \R_+^d$ as well as two indices $u,v\in[d]$. Let $y \in \R_+^d$ be the vector with:
\[ y_k =\left\{ \begin{array}{cc} x_k & k \in [d] \setminus \{ u,v \} \\
						x_u + x_v & k = u \\
						0 		& k = v \end{array} \right. . \] 
Then $\|y\|_X\ge \|x\|_X$.
\end{lemma}

\begin{proof}
Consider the vector $z \in B_{X^*}$ where $\langle x, z\rangle = \|x\|_{X}$. Now, we let $z' \in \R^d$ be given by
\[ z_k' = \left\{ \begin{array}{cc} z_k & k \in [d] \setminus \{ u, v\} \\
						 \max\{ z_u, z_v \} & k = u \\
						 \min\{z_u, z_v\}  & k = v \end{array} \right.\]
Note that $z'$ is a permutation of $z$, so $z' \in B_{X^*}$. Now, 
\begin{align*}
\langle y, z' \rangle = (x_u + x_v) \max\{ z_u, z_v \} + \sum_{k \in [d] \setminus\{u, v\}} x_k z_k \geq \sum_{k \in [d]} x_k z_k = \langle x, z \rangle = \|x\|_{X}.
\end{align*}
\end{proof}

\begin{definition}
Consider a vector $x\in \calR$. We define the {\em simplified rounded vector} $S(x)$ as the vector returned by the following procedure.
\begin{algorithmic}[1]
\State Initialize $z = x$
\For{$k=0, 1, \ldots, 2\log_{\beta}(d)-1$}
\If{$b_k(z) \le \max_{j < k+3\log_{\beta} (\beta-1)} b_j(z)$} 
\State Set all coordinates of $z$ of value $\beta^{-k}$ to $0$ i.e. set $b_k(z) = 0$.
\EndIf
\EndFor
\State Sort the coordinates of $z$ in non-increasing order and return $z$.
\end{algorithmic}

\end{definition}

Next we show that the simplified rounded vector is close to the rounded counts vector.

\begin{lemma}\label{lem:aggregate-net}
Let $\|\cdot\|_X$ be a symmetric norm and let $x\in\calR$. Then $\|S(x)-x\|_X \le 2(\beta-1) \|x\|_X$.
\end{lemma}

\begin{proof}
Consider some $k \in [2\log_{\beta}d - 1]$ and let $C_{k}(x) \subset [d]$ be set of coordinates where $x$ is at level $k$ and does not equal $S(x) = z$, i.e.,
\[ C_{k}(x) = \{ i \in [d] \mid x_i = \beta^{-k} \text{ and } x_i \neq z_i \}. \]
Additionally, for each $k \in [2\log_{\beta} d - 1]$, let $T_k \subset [d]$ be the coordinates at level $k$ in $x$ which trigger line 3 of $S(x)$, and thus become 0's in $z$ (we do not need to consider the case $k = 0$ since line 3 never triggers, in fact, we do not need to consider $k \in [-3\log_{\beta}(b-1)]$ either). In other words,
\[ T_k(x) = \{ i \in [d] \mid x_i = \beta^{-k} \text{ and at iteration $k$ of } S(x), b_k(z) \leq \max_{j < k + 3\log_{\beta}(\beta - 1)} b_j(z) \}. \]
Note that $T_1(x), \dots, T_{2\log_{\beta}d-1}(x)$ are all disjoint, and $|C_{k}(x)| \leq \sum_{j \in [k]} |T_j(x)|$, since whenever we zero out coordinates in levels less than or equal to $k$, $S(x)$ will shift the coordinates when we sort, causing $x_i \neq z_i$. Thus, we may consider an injection $s_{k} \colon C_k(x) \to \bigcup_{j \in [k]} T_j(x)$, which charges coordinates in $C_k(x)$ to coordinates which were zeroed out in line 3 of $S(x)$. 

Additionally, for each $j \in [2\log_{\beta} d - 1]$ where $T_j(x) \neq \emptyset$, we let $q_j$ be the integer between $0$ and $j + 3\log_{\beta}(\beta - 1)$ which triggered line 3 of $S(x)$ at $k = j$. More specifically, $0 \leq q_j \leq j + 3\log_{\beta}(\beta - 1)$ is the integer for which $b_{j}(x) \leq b_{q_j}(x)$.

Finally, for each $j \in [2\log_{\beta}d-1]$ where $T_j(x) \neq \emptyset$, we let $g_j \colon T_j(x) \to B_{q_j}(x)$ (recall that $B_{q_j}(x) \subset [d]$ are the indices of $x$ at level $q_j$) be an arbitrary injection. Such an injection exists because $b_{j}(x) \leq b_{q_j}(x)$.
We may consider the mapping $F \colon \bigcup_{k \in [2\log_{\beta}d - 1]} C_k(x) \to [d]$ where
\[ F(i) = g_j(s_k(i)) \qquad \text{where $k$ and $j$ are such that } i \in C_{k}(x) \text{ and } s_k(i) \in T_{j}(x). \]
See Figure~\ref{fig:charging} for an example of a coordinate being mapped by $F$. 
Let $y$ be the vector where we ``aggregate'' coordinates of $\bigcup_{k \in [2\log_{\beta}(d) - 1]} C_k(x)$ of $x$ according to the map $F$ according to Lemma~\ref{lem:sum}. In particular, we define $y \in \R^d$ where for $i \in [d]$, we let
\[ y_i = \sum_{i' \in F^{-1}(i)} x_{i'}. \]
Note that for each $i \in [d]$, $0 \leq (x - z)_i \leq x_i$, and $\bigcup_{k \in [2\log_{\beta}(d) - 1]} C_k(x) \subset [d]$ are the non-zero coordinates of $x - z$. Thus, from Lemma~\ref{lem:sum}, we conclude that $\|x - z\|_{X} \leq \|y\|_{X}$. We now turn to upper-bounding $\|y\|_{X}$. 

Fix some $i \in [d]$ where $x_i = \beta^{-j}$. Then
\begin{align*}
y_i &= \sum_{k > j - 3\log_{\beta}(\beta - 1)} \left(\sum_{k' \geq k} x_{s_{k'}^{-1}(g_k^{-1}(i))} \right),
\end{align*}
where we interpret $x_{s_{k'}^{-1}(g_k^{-1}(i))}$ as $0$ when $g_{k}^{-1}(i) = \emptyset$, or $s_{k'}^{-1}(g_k^{-1}(i)) = \emptyset$. Note that whenever $x_{s_{k'}^{-1}(g^{-1}_k(i))} \neq 0$, $x_{s_{k'}^{-1}(g^{-1}_k(i))} = \beta^{-k'}$. Thus, 
\begin{align*}
y_i &\leq \sum_{k > j - 3\log_{\beta}(\beta - 1)} \left(\sum_{k' \geq k} \beta^{-k'}\right) \leq \sum_{k > j - 3\log_{\beta}(\beta - 1)} \dfrac{\beta^{1 - k}}{\beta - 1} \leq \dfrac{\beta^{1 - j + 3\log_{\beta}(\beta - 1)}}{(\beta - 1)^2} \leq \beta (\beta - 1)\cdot\beta^{-j}.
\end{align*}
Recall that $x_i = \beta^{-j}$, so $\beta(\beta - 1) x$ weakly majorizes $y$, and thus 
\[ \|x - z\|_{X} \leq \|y\|_{X} \leq \beta(\beta - 1) \cdot\|x\|_{X}. \]
Hence, when $\beta \leq 2$, we have $\|x - z\|_{X} \leq 2 (\beta - 1) \|x\|_{X}$.
\begin{comment}
Thus, we can charge the non-zeros in $x-S(x)$ to the zeroed out coordinates so that each zeroed out coordinate in level $j$ accounts for at most one difference involving a coordinate of $x$ in each level $k\ge j$. For each $j$, let $s_j < j +\log_{\beta} (\beta-1)$ be the level such that $(\beta-1)^2 b_{s_j}(z) \ge b_j(z)$. We deliver all charges to zeroing out level $j$ to level $s_j$. All charges to the same level $k$ are from levels at least $k-\log_{\beta} (\beta-1)$ and they sum up to at most
\begin{align*}
\sum_{j\ge k-\log_{\beta} (\beta-1)} (\beta-1)^2 b_k(z) \sum_{i\ge j} \beta^{-i} &\le \sum_{j\ge k-\log_{\beta} (\beta-1)} (\beta-1)^2 b_k(z) \frac{\beta^{1-j}}{\beta-1}\\
&\le (\beta-1)b_k(z) \beta^{2-k}
\end{align*}
Thus we can divide the charges to the same level $k$ into $b_k(z)$ coordinates as equally as possible and each coordinate  accounts for at most $(\beta-1)\beta^{2-k} + (\beta-1)\beta^{-k}$ (the average charge plus the maximum charge). By Lemma~\ref{lem:sum}, aggregating these charges create a vector $y$ with $\|y\|_X \ge \|x-S(x)\|_X$. However, $y$ is also majorized by $(\beta-1)(\beta^2+1)x$ so $\|x-S(x)\|_X \le (\beta-1)(\beta^2+1)\|x\|_X \leq 5(\beta -1) \|x\|_{X}$ since $\beta \leq 2$.
\end{comment}
\end{proof}

\begin{figure}\label{fig:charging}
\centering
\begin{picture}(370, 70)
\put(0, 0){\includegraphics[width=0.9\linewidth, height=2.5cm]{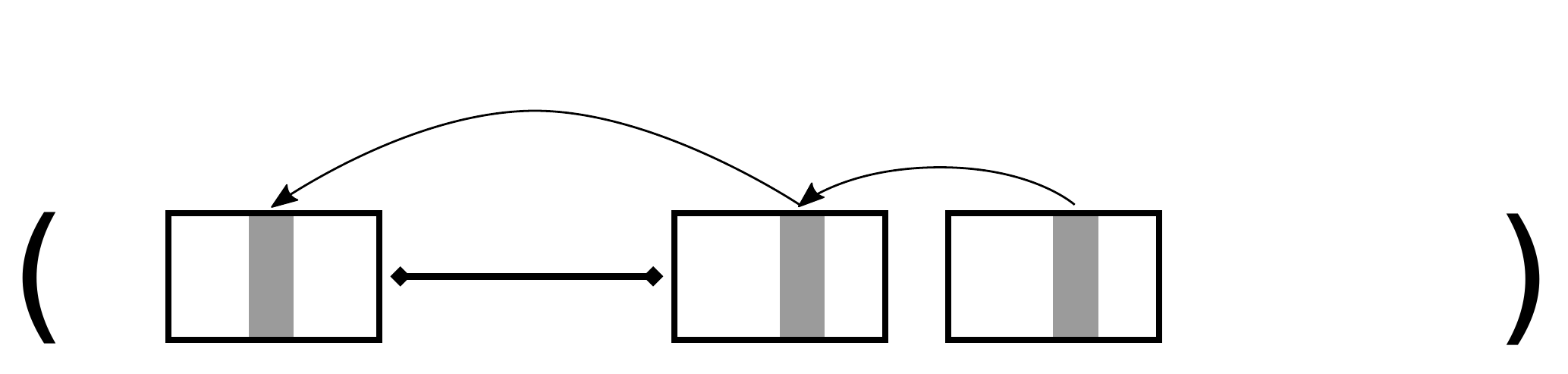}}
\put(-20, 15){$x = $}
\put(60, -7){$B_{q_j}(x)$}
\put(70, 15){$i$}
\put(108, 25){$-3\log_{\beta}(\beta - 1)$}
\put(213, 15){$\ell$}
\put(200, -7){$T_j(x)$}
\put(288, 15){$i'$}
\put(275, -7){$C_k(x)$}
\put(245, 43){$s_k$}
\put(140, 56){$g_j$}
\end{picture}
\caption{Example mapping of particular levels of $x$ with $F$ showing aggregation of coordinates. The coordinate $i' \in C_k(x)$ belongs to level $k$ and will be non-zero in $x - S(x)$. In particular, coordinate $i' \in C_k(x)$ is mapped via $s_k$ to coordinate $\ell \in T_j(x)$, which is $\beta^{-j}$ in $x$ and 0 in $S(x)$. Then coordinate $\ell$ is mapped to $i \in B_{q_j}(x)$, where $q_j$ is the level below $j$ with $b_{q_j} \geq b_j$. Thus, the composed map $F$ sends $i'$ to $i$.}
\end{figure}
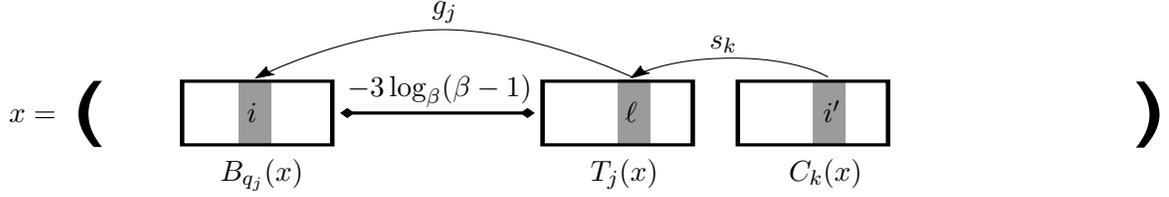

\begin{proof}[Proof of Lemma~\ref{lem:net}]
We now prove the theorem by showing that the set
\[ N = \{ S(x) \in \R^d \mid x \in \calR \} \]
is a $\gamma$-net of $\calR$, and we give an upper bound on the size.  By Lemma~\ref{lem:aggregate-net}, $\|x - S(x)\|_{X} \leq 2(\beta-1) \|x\|_{X} \leq 8(\beta-1)$. So it suffices to give an upper bound on the size of $N$.

We bound from above the number of net points by an encoding argument. Let $z=S(x)$ and let 
\[ t_k = \dfrac{b_k(z)}{\mathop{\max}_{j < k+3\log_{\beta}(\beta-1)} b_j(z)}. \]
 Let $k^* \in \{ 0, \dots, 2\log_{\beta} d - 1\}$ be the smallest level $k$ with non-zero $b_k(z)$. For all $j>k^*-3\log_{\beta}(\beta-1)$, we either have $t_j(z)=0$ or $t_j(z)\ge 1$. Additionally, $z$ has $d$ coordinates, so 
\[ \prod_{j=k^*-3\log_{\beta}(\beta-1)}^{2\log_{\beta}d-1} \max(t_j, 1) \le d^{-3\log_{\beta}(\beta - 1)}, \] 
since terms of the product ``cancel'' except for at most $-3\log_{\beta}(\beta - 1)$, which are each at most $d$. 

We will encode $z \in N$ in three steps. In the first step, we use $2 \log_{\beta}d -1$ bits in order to encode whether $b_{k}(z) = 0$ or $b_{k}(z) > 0$. In the second step, we then encode the values $b_{k^* + j}(z)$ for $j \in \{0, \dots, 3\log_{\beta}(1/(\beta - 1))\}$. Finally, we go through $j > k^* + 3\log_{\beta}\left( \frac{1}{\beta - 1}\right)$, and encode $t_i$ using a prefix-free code, where writing $t_i$ uses at most $O\left(\log \max(t_i, 1) \right)$ many bits.
Thus, in total, the number of bits we use is
\begin{align*} 
&O\left( \log_{\beta} d + \log d \log_{\beta}\left( \frac{1}{\beta - 1} \right) + \sum_{j=k^*-3\log_{\beta}(\beta - 1)}^{2\log_{\beta}d-1} \log \max(t_j, 1)\right) \\
&\qquad= O\left( \dfrac{\log d \cdot \log\left( \frac{1}{\beta - 1}\right)}{\log \beta} + \log\left(\prod_{j=k^* - 3\log_{\beta}(\beta - 1)}^{2\log_{\beta}d - 1} \max(t_j, 1) \right)\right) \\
&\qquad = O\left( \dfrac{\log d \cdot \log\left( \frac{1}{\beta - 1}\right)}{\log \beta}\right).
\end{align*}
Thus, we obtain $N$ is a $8(\beta - 1)$-net, and the size of $N$ is $d^{O(\log(1/(\beta - 1))/\log\beta)}$.
\end{proof}

\section{Proof of the main theorem: ANN for symmetric norms}
% !TEX root = main.tex

We now prove our main result, Theorem~\ref{main_thm}.  The algorithm here
achieves approximation
$$O\left(\frac{\log^2\log n\cdot \log \log d}{\eps^5}\right).$$ We proceed by giving an algorithm for $\bigoplus_{\ell_{\infty}}^t \bigoplus_{\ell_1}^d T^{(c)}$ using
Theorem~\ref{thm:l-infty-ds}, Theorem 5.1.2 from \cite{A09}, and
Theorem~\ref{thm:maxProduct}.

% !TEX root = improve-symm.tex
\begin{lemma}
\label{lem:l-inftykyfan}
Fix some $c_1, \dots, c_d \geq 0$. Let $\bigoplus_{\ell_{\infty}} T^{(c)}$ be the space with $\|\cdot \|_{T, \infty}^{(c)} \colon \R^{d^2} \to \R$ seminorm where for every $x = (x_1, \dots, x_d) \in \R^{d^2}$,
\[ \|x\|_{T, \infty}^{(c)} = \max_{k \in [d]} c_k\|x_k\|_{T(k)}. \] 
For every $\eps \in (0, 1/2)$, there exists a data structure for ANN over $\|\cdot \|_{T, \infty}^{(c)}$ which achieves approximation $O\left( \frac{\log \log n \cdot \log \log d}{\eps^3} \right)$ using space $O\left( d^2\cdot n^{1 + \eps}\right)$ and query time $O\left(d^2\cdot n^{\eps}\right)$.
\end{lemma}

\begin{proof}
Given the randomized embedding from Lemma~\ref{lem:kyfanmap}, we can build a data structure for $c_k\|\cdot \|_{T(k)}$ achieving approximation $O(\frac{\log \log d}{\eps^2})$ using space $O(d^2n^{1 + \eps/2})$ and query time $O(d^2n^{\eps/2})$. This data structure works in the same way as in the proof of Theorem~\ref{thm:dsorlicz}. We handle the constant $c_k$ by rescaling the norm, and since the embeddings are linear, it does not affect the correctness of the data structure. Then we apply Theorem~\ref{thm:l-infty-ds}.
\end{proof}

\begin{lemma}
\label{lem:l-1-kyfan}
Fix some $c_1, \dots, c_d \geq 0$. Let $\bigoplus_{\ell_1} T^{(c)}$ be the space with $\| \cdot \|_{T, 1}^{(c)} \colon \R^{d^2} \to \R$ seminorm where $x = (x_1, \dots, x_m) \in \R^{d^2}$, 
\[ \|x\|_{T, 1}^{(c)} = \sum_{k=1}^d c_k \|x_k\|_{T(k)}. \]
For every $\eps \in (0, 1/2)$, there exists a data structure for ANN over $\|\cdot \|_{T,1}^{(c)}$ which achieves approximation $O\left( \frac{\log \log n \cdot\log\log d}{\eps^4} \right)$ using space $O(d^2\cdot n^{1+\eps})$ and query time $O(d^2\cdot n^{\eps})$. 
\end{lemma}

\begin{proof}
The proof follows from Theorem 5.1.2 in \cite{A09} and Lemma~\ref{lem:l-inftykyfan}.
\end{proof}

Finally, we are combine the above results to get an improved algorithm
for general symmetric norms.

\begin{theorem}\label{thm:main-formal}
For every $d$-dimensional symmetric norm $\|\cdot\|_{X}$ and every
$\eps \in (0, 1/2)$, there exists a data structure for ANN over
$\|\cdot \|_{X}$ which achieves approximation $O(\frac{\log^2 \log
n \log \log d}{\eps^5})$ using space $d^{O(1)} \cdot O(n^{1
+ \eps})$ and query time $d^{O(1)} \cdot O(n^{\eps})$.
\end{theorem}

\begin{proof}
Given Theorem~\ref{thm:symnorm}, we embed $\|\cdot\|_{X}$ into $\bigoplus_{\ell_{\infty}} \bigoplus_{\ell_1} T^{(c)}$ with approximation $(1 \pm \frac{1}{10})$. The result from Lemma~\ref{lem:l-1-kyfan} allows we to apply Theorem~\ref{thm:maxProduct} to obtain the desired data structure.
\end{proof}

Theorem~\ref{thm:main-formal} implies our main result
Theorem~\ref{main_thm} stated in the introduction.

\section{Lower bounds}
% !TEX root = main.tex

\label{sec:linearLB}

In this section, we show that our techniques do not extend to general
norms. In particular, we show there does not exist a \emph{universal}
norm $U$ for which any norm embeds (possibly randomized) with constant
distortion, unless the blow-up in dimension is exponential. Hence the
result from below applies to cases of $U=\ell_\infty$ as well as an
(low-dimensional) product spaces.

\begin{theorem}
\label{thm:linearLB}
For any $\eps > 0$, let $U$ be a $d'$-dimensional normed space such that for any $d$-dimensional normed space $X$, there exists a distribution $\calD$ supported on linear embeddings $f \colon \R^d \to \R^{d'}$ where for every $x \in \R^d$, 
\[ \|x\|_{X} \leq \|f(x)\|_{U} \leq D \|x\|_{U} \]
holds with probability at least $\frac{2}{3}$ over the draw of $f \sim \calD$, for $D = O(d^{1/2 - \eps})$. Then $d' = \exp\left( \Omega(d^{2\eps})\right)$. 
\end{theorem}

We will prove the above theorem by showing that if there exists a universal normed space $U$ satisfying the conditions of Theorem~\ref{thm:linearLB} above, then two parties, call them Alice and Bob, can use the embeddings to solve the communication problem \textsc{Index} with only a few bits. Let $U$ be a proposed $d'$-dimensional normed space satisfying the conditions of Theorem~\ref{thm:linearLB}. By the John's theorem \cite{B97a}, we may apply a linear transform so that:
\[ B_{\ell_2} \subset B_U \subset \sqrt{d'} B_{\ell_2} \]

\begin{lemma}
For any $\eps > 0$, there exists a set of $\exp\left( \Omega(d^{2\eps}) \right)$ many points on the unit sphere $S^{d-1}$ such that pairwise inner-products are at most $\frac{1}{d^{1/2 - \eps}}$. In fact, these points may consist of points whose coordinates are $\pm \frac{1}{\sqrt{d}}$. 
\end{lemma}

\begin{proof}
Consider picking two random points $x, y \in S^{d-1}$ where each entry is $\pm \frac{1}{\sqrt{d}}$. Then by Bernstein's inequality,
\[ \Pr_{x, y} \left[ |\langle x, y \rangle| \geq \frac{1}{d^{1/2 - \eps}} \right] \leq 2\exp\left( -\Omega(d^{2\eps}) \right) \]
We may pick $\exp\left( \Omega(d^{2\eps})\right)$ random points and union bound over the probability that some pair has large inner product.
\end{proof}

Fix $\eps > 0$ and $C = d^{1/2 - \eps}$, and let $P$ be set a set of unit vectors with pairwise inner-product at most $\frac{1}{C}$ of size $\exp(\Omega(d^{2\eps}))$. For each $a \in \{0, 1\}^{P}$ consider the following norm:
\[ \|x\|_{a} = C\cdot \max_{y \in P : a_y = 1} |\langle x, y \rangle|. \]
Assume there exists a randomized linear embedding $f \colon \R^d \to \R^{d'}$ with the following guarantees: 
\begin{itemize}
\item For every $x \in \R^d$, $\|x\|_{a} \leq \|f(x)\|_{U} \leq D \|x\|_{a}$ with probability at least $\frac{2}{3}$. 
\end{itemize}

Note the embedding $f$ can be described by $M$, a $d' \times d$ matrix of real numbers. Additionally, we consider rounding each entry of $M$ by to the nearest integer multiple of $\frac{1}{\poly(d)}$ to obtain $M'$. For each $x \in S^{d-1}$, $\|(M - M') x \|_{U} \leq \|(M - M') x\|_{2} \leq \frac{1}{\poly(d)}$. Thus, we may assume each entry of $M$ is an integer multiple of $\frac{1}{\poly(d)}$, and lose $(1 \pm \frac{1}{\poly(d)})$ factor in the distortion of the embedding for vectors in $B_2$. 

We now show that the existence of the randomized embedding implies a one-way randomized protocol for the communication problem \textsc{Index}. We first describe the problem. In an instance of \textsc{Index}:
\begin{itemize}
\item Alice receives a string $a \in \{0, 1\}^n$. 
\item Bob receives an index $i \in [n]$. 
\item Alice communicates with Bob so that he can output $a_i$. 
\end{itemize}

\begin{theorem}[\cite{KNR99}]
\label{thm:indexing}
The randomized one-way communication complexity of \textsc{Index} is $\Omega(n)$. 
\end{theorem}

\def\span{\text{span}}

We give a protocol for \textsc{Index}:
\begin{enumerate}
\item Suppose Alice has input $a \in \{0, 1\}^P$. She will
  generate the norm $\|\cdot\|_a$ described above. Note that $f \sim \calD$ has that for each $x \in \R^d$, the embedding preserves the norm of $x$ up to $D$ with probability $\frac{2}{3}$. In particular, if Bob's input is $i \in |P|$, corresponding to point $y$, then an embedding $f \in \calD$, which we represent as a $d' \times d$ matrix $M$, satisfies: 
 \[\|y\|_a \leq \|My\|_{U} \leq D\|y\|_{a} \]
 with probability $\frac{2}{3}$. In particular, with probability $\frac{2}{3}$:
\begin{itemize}
\item If $a_i = 0$, then $\|y\|_a \leq 1$, which implies $\|My\|_{U} \leq D$.
\item If $a_i = 1$, then $\|y\|_a \geq C$, which implies $\|My\|_{U}\geq C$. 
\end{itemize}
Alice computes the set $P_{c} \subset P$ of vectors which satisfy the above property (i.e. the embedding $M$ preserves increases the norm by at most a factor $D$).
\item Alice finds a subset $B \subset P_c$ of linearly independent
vectors such that every $x \in P_c$ we have $x \in \span(B)$. Note
that $|B| \leq d$ and for all $x \in B$,
$\|Mx\|_{2} \leq \sqrt{d'} \|Mx\|_{U} \leq C\cdot
D \cdot \sqrt{d'}$. Therefore, each $Mx \in \R^{d'}$ can be written
with $\tilde{O}(d')$ bits. So Alice sends the set $B$, as well as
$Mx$ for each $x \in B$ using $\tilde{O}(dd')$ bits.
\item In order for Bob to decode $a_i$, he first checks whether $y \in \span(B)$, and if not, he guesses. If $y \in \span(B)$, which happens with probability $\frac{2}{3}$, then Bob writes
\[ y = \sum_{b_i \in B} c_i b_i \]
and $My = \sum_{b_i \in B} c_i Mb_i$. If $\|My\|_{U} \leq D$, then $a_i = 0$ and if $\|My\|_{U} \geq C$ then $a_i = 1$. Thus, if $D < \frac{C}{2}$, Bob can recover $a_i$ with probability $\frac{2}{3}$. 
\end{enumerate}

Alice communicates $\tilde{O}(dd')$ bits, and Bob is able to recover $a_i$ with probability $\frac{2}{3}$. By Theorem~\ref{thm:indexing}, $dd' \geq \tilde{\Omega}\left(|P|\right)$, which in turn implies $d' \geq \exp\left( \Omega(d^{2\eps}) \right)$.

\section*{Acknowledgments}

We thank Piotr Indyk for suggesting the proof of
Theorem~\ref{thm:linearLB}. We also thank Assaf Naor for discussions
related to these research questions. Thanks to Cl\'{e}ment Canonne
for pointing us to relevant literature about symmetric norms.

This work is supported in part by Simons Foundation, Google, and NSF
(CCF-1617955), as well as by the NSF Graduate Research Fellowship
(Grant No. DGE-16-44869).

{\small
\bibliography{bibfile}
\bibliographystyle{alpha}
}

\appendix

\section{Bounding space in Theorem \ref{thm:maxProduct}}
\label{apx:space}

Here we justify the space bound of the algorithm from Theorem
\ref{thm:maxProduct} (from \cite{I02}). We note that the improved
bound was also claimed in \cite{AIK09}, albeit without a proof.

First of all, as suggested at the end of Section~3 of \cite{I02}, one
modifies the algorithm to obtain space of the form of
$n^{1+\eps}$, at the expense of increasing the approximation to
$O(\eps^{-1}\log\log n)$. This is done by replacing the conditions in
Case 2 and 3 by respectively:
$$
\left[
\frac{|B(s,R(s)+c+1)\cap S_{|i}|}{|S_{|i}|}
\right]^{1+\eps}
<
\frac{|B(s,R(s))\cap S_{|i}|}{|S_{|i}|},
$$
and
$$
\left[
\frac{|S_{|i}-B(p,R')|}{|S|}
\right]^{1+\eps}
<
\frac{|S_{|i}-B(s,R'+2)|}{|S|}.
$$

With the remaining algorithm being precisely the same, our only task
here is to argue the space bound. First of all we bound {\em the sum
  of the number of points stored in all the leaves}. For a tree with
$m$ nodes, let $L(m)$ be an upper bound on this count. We would like
to prove that $L(m)\le m^{1+\eps}$.  As in \cite{I02}, we only need
to focus on cases 2 and 3 of the construction, as case 1 does not
replicate the points.  We will consider the case 2 (case 3 is exactly
similar). 

Let $m_j=|S_j|$ and $m_j'=|S_{|i}\cap \cup_{s\in S_j}B(s,c+1)|$,
whereas $|S|=m$. By construction, we have that $\sum m_j=m$ and
$m_j/m>(m_j'/m)^{1+\eps}$ for all $j$.

By induction, assume $L(m_j')\le (m'_j)^{1+\eps}$ for all
children. Then, we have that:
$$
L(m)\le \sum_j L(m_j')\le \sum_j (m_j')^{1+\eps}<
m^\eps\sum_j m_j=m^{1+\eps}.
$$

We now argue the total space is $O(S(n)\cdot k\log n\cdot
n^{\epsilon})$.  Since the depth of the tree is $O(k\log n)$, we have
that the total number of points stored in the ANN data structures is
$O(k\log n\cdot C(n))=O(k\log n\cdot n^{1+\epsilon})$. Since each ANN
is on at most $n$ points, we have that, for each occurrence of a point
in the ANN data structure, we have an additional factor of
$S(n)/n$.\footnote{Here we assume the natural condition that $S(n)$ is
  increasing, which is, otherwise, easy to guarantee.} Hence the total
space occupied by all the ANN data structures is $O(S(n)/n\cdot k\log
n\cdot n^{1+\epsilon})$. Using a smaller $\eps$ (to hide the $\log n$
factor), we obtain the stated space bound of $O(S(n)\cdot k\cdot n^{\epsilon})$.

\section{$\widetilde{O}(\log d)$-ANN for symmetric norms}
\label{sec:simple-sym}
% !TEX root = main.tex
We provide a simple ANN algorithm for general symmetric norm achieving $O(\log d \log \log d)$ approximation using near-linear space and sub-linear query time. The algorithm will leverage the results in the previous section by relating general symmetric norms to Orlicz norms. Recall the definition of level vectors in Definition~\ref{def:levels}.

\begin{definition}
Let $\|\cdot \|_{X}$ be any symmetric norm. Let $L_k > 0$ be the minimum number of coordinates needed at level $k$ to have norm at least $1$. In other words,
\[ L_k = \min \{ j \in [d] \mid \| \beta^{-i} \xi^{(j)} \|_{X} > 1 \}. \]
\end{definition}

At a high level, we will relate the norm of a vector $x \in \R^d$ to the norm of its level vectors $V_k(x)$. The definition above gives a way to measure the contribution of level $k$ to the norm. For example, if $x \in \R^d$ has norm $\|x\|_{X} \geq D$, and there are only $2\log_\beta d$ non-zero levels with respect to $x$, then some level vector $\|V_k(x)\|_{X} \geq \frac{D}{2\log_\beta d}$. This implies $b_k = \Omega(\frac{DL_k}{\log_\beta d})$, since we may divide $V_k(x)$ into a sum of vectors with $L_k$ coordinates at level $k$. 

On the other hand, if $x \in \R^d$ has $\|x\|_{X} \leq 1$, then $b_k < L_k$ for each $k$. Since we consider only $2 \log_\beta d$ relevant levels, for $\|x\|_{S} \leq 1$,
\[ \sum_{k=0}^{2\log_{\beta} d-1} \dfrac{b_k}{L_k} \leq 2 \log_{\beta} d. \] 
Additionally, $\sum_{k=0}^{2\log_{\beta} d-1} (b_k/L_k)$ can be decomposed as an additive contribution of coordinates. In particular, coordinate $x_i$ contributes $1/L_k$ if  $i \in B_k$. Therefore, we can hope to approximate the symmetric norm by an Orlicz norms and apply the arguments from Lemma~\ref{lem:orliczmap}. 

The lemma below formalizes the ideas discussed above.

\begin{lemma}
Let $\|\cdot \|_{X}$ be any symmetric norm. For any $D, \alpha > 1$, there exists a non-decreasing function $G \colon \R_+ \to \R_+$ with $G(0) = 0$ and $G(t) \to \infty$ as $t \to \infty$, where every vector $x \in \R^d$ satisfies the following:
\begin{itemize}
\item If $\|x\|_{X} \leq 1$, then $\sum_{i=1}^d G(|x_i|) \leq 2\log_{\beta} d$. 
\item If $\|x\|_{X} > \alpha D \cdot 7 \log_{\beta} d$, then $\sum_{i=1}^d G\left(\frac{|x_i|}{D}\right) \geq \alpha \cdot 2\log_{\beta} d$. 
\end{itemize}
\end{lemma}

\begin{proof}
For $i \geq 0$, let $A_i = (\beta^{-i-1} , \beta^{-i}]$. The function $G \colon \R_+ \to \R_+$ is defined as
\begin{equation}
\label{eq:g}
 G(t) = \sum_{i = 0}^{2\log_\beta d - 1} \dfrac{\chi_{A_i}(t)}{L_i} + \alpha \cdot 2 \log_{\beta} d \cdot t \cdot \chi_{(1, \infty)}(t)
 \end{equation}
Note that $G(0) = 0$ and $G(t) \to \infty$ as $t\to \infty$. 

Recall the norm satisfies, $\|\xi^{(1)}\|_{X} = 1$, so if $\|x\|_{X} \leq 1$, then $|x_i|\leq 1$ for all $i \in [d]$. This means $\chi_{(1,\infty)}(|x_i|) = 0$ so the second term of the RHS of (\ref{eq:g}) is zero. Therefore,
\[ \sum_{i=1}^{d} G(|x_i|) = \sum_{i=1}^d \sum_{k=0}^{2\log_\beta d - 1} \dfrac{\chi_{A_k}(|x_i|)}{L_k} = \sum_{k=0}^{2\log_\beta d - 1} \dfrac{b_k}{L_k} \]
where $b_k$ is defined with respect to $x$. Since, $b_k < L_k$ for all $0 \leq k < 2\log_{\beta} d$,
\[ \sum_{i=1}^d G(|x_i|) \leq 2\log_{\beta} d. \]
If $x \in \R^d$ where $\|x\|_{X} > \alpha D\cdot 7 \log_{\beta} d$, then the vector $\|\frac{x}{D}\|_{X} > \alpha \cdot 7 \log_{\beta} d$. So it suffices to prove that for any vector $x \in \R^d$ with $\|x\|_{X} > \alpha \cdot 7 \log_{\beta} d$,
\[ \sum_{i=1}^d G(|x_i|) \geq \alpha \cdot 2\log_{\beta} d \]
Additionally, for any vector $x \in \R^d$, we may consider the vector $C(x) \in \R^d$ for $\tau = \frac{\beta}{d^2}$ from Definition~\ref{def:cut-small}. By Proposition~\ref{prop:cut-small-level}, $\|C(x)\|_{X} \geq \|x\|_{X} - \frac{\beta}{d} > \alpha \cdot 6 \log_{\beta} d$. Therefore, we may assume $x \in \R^d$ has $\|x\|_{X} > \alpha \cdot 6 \log_{\beta} d$, and that all non-zero coordinates have absolute values greater than $\frac{\beta}{d^2}$. Equivalently, $b_k = 0$ for all $k \geq 2\log_{\beta} d$. If for some $i \in [d]$, $|x_i| \geq 1$, then the second term in the RHS of (\ref{eq:g}) is non-zero, and $G(|x_i|) \geq \alpha \cdot 2 \log_{\beta} d$. So we may further assume all coordinates of $x$ lie in levels $k = 0, \dots, 2\log_{\beta} d - 1$. Note that 
\[ \sum_{i=1}^d G(|x_i|) = \sum_{k=0}^{2\log_{\beta} d - 1} \sum_{i=1}^d G(|V_k(x)_i|), \]
and for each $0 \leq k < 2\log_{\beta} d$, $\sum_{i=1}^d G(|V_k(x)_i|) = \dfrac{b_k}{L_k}$.

We partition the levels into two groups, 
\[ A = \left\{ k \mid \frac{b_k}{L_k} < 1 \right\} \qquad \text{ and } \qquad B = \left\{ k \mid \frac{b_k}{L_k} \geq 1\right\} . \]
For all $k \in B$,
\[ \|V_k(x)\|_{X} \leq \left\lceil \dfrac{b_k}{L_k}\right\rceil \leq \frac{2b_k}{L_k} \]
since by the triangle inequality, we can break $V_k(x)$ into at most $\left\lceil \dfrac{b_k}{L_k} \right\rceil$ vectors with $L_k$ coordinates at level $k$ each having norm at least $1$. 

Suppose for the sake of contradiction that $\sum_{k \in B} \frac{b_k}{L_k} \leq \alpha \cdot 2 \log_{\beta} d$. Then
\[ \alpha \cdot 4 \log_{\beta} d \geq \sum_{k \in B} \frac{2b_k}{L_k} \geq \sum_{k \in B} \|V_k(x)\|_{X}. \]
Additionally, since $\|x\|_{X} > \alpha \cdot 6 \log_{\beta} d$, and 
\[ \alpha \cdot 6 \log_{\beta} d < \|x\|_{X} \leq \sum_{k \in A} \|V_k(x)\|_{X} + \sum_{k \in B}\|V_k(x)\|_{X},\]
it follows that 
\[ \sum_{k \in A} \|V_k(x)\|_X > \alpha \cdot 2\log_{\beta} d. \]
However, this is a contradiction for since $|A| \leq 2 \log_{\beta} d$ and $\|V_k(x)\|_{X} \leq 1$.
\end{proof}

\begin{lemma}
For any $\eps \in (0, 1/2)$, there exists a data structure for ANN over any symmetric norm $\|\cdot\|_{X}$ which achieves approximation $O\left(\frac{\log d \log \log d}{\eps^2}\right)$ using space $O(dn^{1 +\eps})$ and query time $O(dn^{\eps})$.
\end{lemma}

\begin{proof}
We fix $\beta = \frac{3}{2}$. The proof of this lemma follows in the same way as the proof of Theorem~\ref{thm:dsorlicz}. The one difference is that we rescale the $\ell_{\infty}$ norm by $\frac{1}{2 \log_{\beta} d}$ after applying the embedding. 
\end{proof}

\subsection{The $\log^{\Omega(1)} d$-approximation is necessary}
\label{logd_lower}

Let us remark that we cannot push the technique much further. Namely, any $G(\cdot)$ (even non-convex) requires approximation
$\Omega(\sqrt{\log d})$ for the following norm. Define the norm of a vector to be
$$
  \|x\| = \max_{1 \leq k \leq d} \left(\frac{x^*_1 + x^*_2 + \ldots x^*_k}{\sqrt{k}}\right).
  $$
  This is the minimal norm for $a_k = \sqrt{k}$ (see Section~\ref{ex_symmetric} for the definition). It is not hard to check that an approximation
  with any $G(\cdot)$ ends up having a distortion $\Omega(\sqrt{\log d})$. 

  The idea is to consider the following vectors: for every $1 \leq k \leq d$, we consider a vector
  $$
  \Bigl(\underbrace{1, 1, \ldots, 1}_{k}, 0, 0, \ldots, 0\Bigr),
  $$
  and besides, we consider a vector
  $$
  \Bigl(1, \sqrt{2} - 1, \sqrt{3} - \sqrt{2}, \ldots, \sqrt{d} - \sqrt{d - 1}\Bigr).
  $$
  The remaining calculation is a simple exercise.

\section{Lower bound for arbitrary metrics: expander graphs}
% !TEX root = main.tex
\label{apx:lbMetrics}

We give an example of a {\em metric} that is hard for current
approaches to ANN search. The lower bound is based on the notion of robust
expansion, which implies all known lower bounds for
ANN \cite{PTW10, ALRW16a}. In what follows, we will refer to $d = \log
N$ as the dimension of a finite metric space of size $N$.

Our example of a hard metric will be the shortest path metric on any
spectral expander graph. We note that a similar theorem to the one
below is also known for a finite subset of
the high-dimensional Earth-Mover Distance \cite{KP12}.

Fix $M$ to be the metric induced by the shortest path distance on a
3-regular expander $G$ on $N$ nodes. In particular, assume that $1 -
\lambda(G) > c$, where $c$ is an absolute constant, and $\lambda(G) \in
(0,1)$ is the second-largest eigenvalue of the normalized adjacency
matrix of $G$. Let $d$ be the dimension $d=\log N$.

\begin{theorem}
\label{thm:cell-probe}
For any approximation $\alpha>1$, and data set size $n\ge 1$ with
$d^{\Omega(1)}\le n\le N^{O(1)}$, any $\alpha$-ANN data structure on
$n$ points which makes $t$ cell probes (with cells of size at most
$w\le (d\log n)^{O(1)}$), and has success probability at least
$\gamma>n^{-1+o(1)}$, must use space
$m=\gamma^{\Omega(1/t)}{N^{\Omega(1/(\alpha
t))}}=\gamma^{\Omega(1/t)}2^{\Omega(d/(\alpha t))}$.
\end{theorem}

We proceed by introducing a few definitions from \cite{PTW10},
and then prove lower bounds on the robust expansion.

\begin{definition}[\cite{PTW10}]
\label{def:gns}
In the Graphical Neighbor Search problem (GNS), we are given a
bipartite graph $H = (U, V, E)$ where the dataset comes from $U$ and
the queries come from $V$. The dataset consists of pairs $P = \{ (p_i,
x_i) \mid p_i \in U, x_i \in \{0, 1\}, i \in [n] \}$. On query $q\in
V$, if there exists a unique $p_i$ with $(p_i, q) \in E$, then we want
to return $x_i$.
\end{definition}

One can use the GNS problem to prove lower bounds on $c$-ANN as
follows: build a GNS graph $H$ by taking $U = V = [N]$, and connecting
two points $u\in U, v\in V$ iff they are at a distance at most $r$
(see details in \cite{PTW10}). We will also need to make sure that in
our instances $q$ is not closer than $cr$ to other points except the
near neighbor.

We now introduce the notion of robust expansion, used in
\cite{PTW10} to prove lower bounds.

\begin{definition}[Robust Expansion \cite{PTW10}]
For a GNS graph $H=(U,V,E)$, fix a distribution $e$ on $E\subset
U\times V$, and let $\mu$ be the marginal on $U$ and $\eta$ be the
marginal on $V$.  For $\delta, \gamma\in(0,1]$, the robust expansion
$\Phi_r(\delta, \gamma)$ is defined as follows:
$$
\Phi_r(\delta, \gamma)=\min_{A\subset V : \eta(A)\le \delta} \min_{B\subset U :
    \frac{e(A\times B)}{e(A\times V)}\ge \gamma} \frac{\mu(B)}{\eta(A)}.
$$
\end{definition}

We now prove a lower bound on the robust expansion $\Phi_r(\delta,
\gamma)$ for a GNS graph arising from the shortest path metric on the
expander graph $G$. Fix $r=d/\alpha$. The hard distribution $e$ is
defined as follows: pick $p$ at random from $M$ and obtain $q$ by
running a random walk of length $r$ starting at $p$. Note that for
$n<N^{1/4}$ and sufficiently high constant $\alpha$, the distribution
satisfies the weak-independence condition required for applying the
results in~\cite{PTW10}.

Fix any sets $A,B\subset M$, where $a=|A|/N$ and $b=|B|/N$. By the
expander mixing lemma applied to $G^r$, we obtain that:
$$
\left|E_{G^r}(A,B)-\tfrac{|A|\cdot |B|}{3^r N}\right|\le \lambda^3
3^r\sqrt{|A|\cdot |B|}.
$$
Considering that $\Pr[q\in B\mid p\in
  A]=\tfrac{E_{G^r}(A,B)}{aN\cdot 3^r}$, we have that:
$$ 
\Pr[q\in B\mid p\in A]\le b+\lambda^r\sqrt{b/a}.
$$ Restricting to sets $A,B$ such that $\Pr[q\in B\mid p\in A]\ge
\gamma$, for which we must have that $\Phi_r=\Phi_r(a,\gamma)\ge b/a$ (by definition),
we conclude:
$$
\gamma\le \Phi_r\cdot a+\lambda^r\sqrt{\Phi_r}.
$$

Hence, either $\Phi_r=\Omega(\gamma/a)$ or $\Phi_r=\Omega(\gamma^2/\lambda^{2r})$. 

\begin{proof}[Proof of Theorem~\ref{thm:cell-probe}]
Applying Theorem 1.5 from \cite{PTW10}, we have that, for $t\ge 1$
cell probes, either:
\begin{itemize}
\item
$m^tw/n\ge \Omega(\gamma\cdot m^t)$, an impossibility;
\item
or
$m^tw/n\ge \Omega(\gamma^2/\lambda^{2r})$, or
$m^t=\Omega(\frac{n}{w}\gamma^2/ \lambda^{2r})$, implying $m=\gamma^{2/t}N^{\Omega(1/(\alpha
  t))}$.
\end{itemize}
\end{proof}

To show how bad the situation for expander metrics is, we state a lower bound on
for $\alpha$-ANN on the expander metric described above in the \emph{list-of-points} model, which captures the hashing-based algorithms of \cite{ALRW16a} and in the decision tree model of \cite{I01}. The proofs follow from a simple derivation using the robust expansion lower bounds in Section~7 of \cite{ALRW16a} and a reduction of decision trees to $O(\log m)$-cell-probe data structures similar to Appendix A in~\cite{ACP08}.

\begin{theorem}
Any list-of-points data structure for $(c, r)$-ANN for random instances of $n$ points in the expander metric of dimension $d$ (described above) with query time $t$ and space $m$ has either $t = \Omega(n)$, or $m = \exp\left( \Omega(d)\right)$.
\end{theorem}

\begin{theorem}
Let $d = \Omega(\log^{1 + \eps} n)$ for some $\eps > 0$. Any decision tree of size $m$ and depth $t$ and word size $w$ succeeding with probability $\gamma$ satisfies:
\[ \dfrac{m^{O(\log m)} t w }{n} \geq \Phi_r\left(\frac{1}{m^{O(\log m)}}, \frac{\gamma}{O(\log m)} \right). \]
In particular, for any $\rho > 0$, if $w \leq n^{\rho}$, either $t \geq \widetilde{\Omega}(n^{1-\rho})$ or $m = \exp\left(\Omega( d^{\eps/(1+\eps)}) \right) \poly(n)$. 
\end{theorem}

\end{document}